\documentclass[]{article}

\usepackage{tikz, pgfplots}
\pgfplotsset{compat=newest}
\pgfplotsset{width = 5.5cm,
			 every axis plot post/.append style = {line width = 1pt,
		 		every mark/.append style = {solid}}}
\usepackage[ruled, vlined]{algorithm2e}
\usepackage[nocompress]{cite}
\usepackage{amsfonts, amsmath, amssymb, amsthm}
\usepackage{mathtools}
\usepackage{xparse}
\usepackage{hyperref, cleveref}

\newtheorem{theorem}{Theorem}[section]
\newtheorem{lemma}[theorem]{Lemma}
\newtheorem{corollary}[theorem]{Corollary}
\newtheorem{proposition}[theorem]{Proposition}
\newtheorem{remark}[theorem]{Remark}
\newtheorem{assumption}[theorem]{Assumption}
\newtheorem{example}[theorem]{Example}
\crefname{assumption}{assumption}{assumptions}

\usepackage{graphicx}
\usepackage{xcolor}
\hypersetup{  
    colorlinks,
    linkcolor={blue!50!black},
    citecolor={green!50!black},
    urlcolor={red!80!black}
}

\makeatletter

\newcommand\@enclose[3]{#1#3#2}
\newcommand\@@enclose[3]{\left#1#3\right#2}
\newcommand\@big@enclose[3]{\big#1#3\big#2}
\newcommand\@bigg@enclose[3]{\bigg#1#3\bigg#2}
\newcommand\@Bigg@enclose[3]{\Bigg#1#3\Bigg#2}
\newcommand\@Big@enclose[3]{\Big#1#3\Big#2}

\DeclareRobustCommand\en{
	\@ifnextchar *
	{\@firstoftwo{\@@enclose}} 
	{\@ifnextchar 4
		{\@firstoftwo{\@Bigg@enclose}}
		{\@ifnextchar 3
			{\@firstoftwo{\@bigg@enclose}}
			{\@ifnextchar 2
				{\@firstoftwo{\@Big@enclose}}
				{\@ifnextchar 1
					{\@firstoftwo{\@big@enclose}}
					{\@enclose}
	}}}}
}

\def\@setdelimsize#1#2{
	\def\@argStar{#1}
	\IfBooleanTF{#1}
	{\def\delimsize{\middle}
		\IfValueT{#2}{\PackageError{MathDelimiters}{Cannot provide both star and size}{Please specify one or the other}}}   
	{\IfValueTF{#2}{\def\delimsize{#2}}{\def\delimsize{}}}}

\def\xParseDeclarePairedDelimiter#1#2#3#4#5#6#7{
	\NewDocumentCommand{#1}{
		#2}
	{
		\begingroup
		\expandafter\IfValueTF\expandafter{#3}{#3}{\@setdelimsize{##1}{##2}}
		\expandafter\IfBooleanTF\expandafter{\@argStar}
		{\mathopen{}\mathclose\bgroup\left#4 #7 \aftergroup\egroup\right#5}
		{\mathopen{\delimsize#4}#7\mathclose{\delimsize#5}}
		\endgroup
	}
}
\def\@space{\nonscript\,}
\providecommand\given{}
\newcommand\@given[1][]{
	\@space#1\vert \allowbreak \@space\mathopen{}}
\newcommand\@suchthat[1][]{
	\@space : \allowbreak \@space\mathopen{}}
\def\xParseDeclareExpectation#1#2#3#4{
	\xParseDeclarePairedDelimiter{#1}
	{e{^_} s o m}
	{\@setdelimsize{##3}{##4}
		\renewcommand\given[1][\delimsize]{\@given[##1]}
		#2\IfValueT{##1}{^{##1}}\IfValueT{##2}{_{##2}}}
	#3#4{}
	{\renewcommand\given{\@given[\delimsize]}\@space##5\@space}}

\xParseDeclareExpectation{\E}{\mathbb{E}}[]
\xParseDeclareExpectation{\prob}{\mathbb{P}}[]
\xParseDeclareExpectation{\var}{\mathrm{Var}}[]
\xParseDeclareExpectation{\cov}{\mathrm{Cov}}[]

\DeclarePairedDelimiterXPP\Order[1]{\mathcal{O}}(){}{#1}
\DeclarePairedDelimiterXPP\order[1]{o}(){}{#1}
\DeclarePairedDelimiter\abs{\lvert}{\rvert}
\DeclarePairedDelimiter\norm{\|}{\|}
\DeclarePairedDelimiter\p{(}{)}

\DeclarePairedDelimiterXPP\set[1]{}\{\}{}{
	\renewcommand\given{\@suchthat}#1}

\newcommand*{\@@eqdef}{\textnormal{\fontsize{4pt}{4pt}\selectfont def}}
\newcommand*{\eqlabelleddef}{\hskip 0.03cm\mathop{\overset{\@@eqdef}{
			\resizebox{\widthof{\ensuremath{\mathop{\overset{\@@eqdef}{=}}}}}{\heightof{=}}{=}}}\hskip
	0.03cm}

\def\defeq{\coloneqq}
\def\eqdef{\eqqcolon}
\def\H#1{\mathbb H\p{#1}}
\def\I#1{\mathbb I_{#1}}

\def\dom#1{d_\Omega\p{#1}}
\def\dbar#1{\bar d_{\Omega}\p*{#1}}
\def\ceil#1{\left \lceil #1 \right \rceil}
\def\tol{\varepsilon}
\def\iid{\overset{\textnormal{i.i.d.}}{\sim}}
\def\dh{\H{g_\ell} - \H{g_{\ell-1}}}
\def\dhm{\H{g_\ell^{(f,m)}} - \H{g_{\ell-1}^{(c,m)}}}
\def\dhl{\Delta\mathbb{H}_\ell}

\def\mlmc{\mathcal{M}^*}
\def\work{\textnormal{Work}\p*{\mlmc; \tol}}

\DeclareMathOperator{\R}{\mathbb{R}}
\begin{document}

\title{Adaptive Multilevel Monte Carlo for Probabilities\thanks{Submitted 15 July, 2021}}
\author{Abdul-Lateef Haji-Ali\thanks{Heriot-Watt University (\href{mailto::A.HajiAli@hw.ac.uk}{A.HajiAli@hw.ac.uk}, \href{mailto::jws5@hw.ac.uk}{jws5@hw.ac.uk}).}
		\and Jonathan Spence\footnotemark[2]
		\and Aretha Teckentrup\thanks{University of Edinburgh (\href{mailto::a.teckentrup@ed.ac.uk}{a.teckentrup@ed.ac.uk})}
	}
\date{}
\maketitle

\begin{abstract}
\textbf{AMS Subject Classication: } 65C05, 62P05


We consider the numerical approximation of $\prob{G\in \Omega}$ where the $d$-dimensional random variable $G$ cannot be sampled directly, but there is a hierarchy of increasingly accurate approximations $\{G_\ell\}_{\ell\in\mathbb{N}}$ which can be sampled. The cost of standard Monte Carlo estimation scales poorly with accuracy in this setup since it compounds the approximation and sampling cost. A direct application of Multilevel Monte Carlo improves this cost scaling slightly, but returns sub-optimal computational complexities since estimation of the probability involves a discontinuous functional of $G_\ell$. We propose a general adaptive framework which is able to return the MLMC complexities seen for smooth or Lipschitz functionals of $G_\ell$. Our assumptions and numerical analysis are kept general allowing the methods to be used for a wide class of problems. We present numerical experiments on nested simulation for risk estimation, where $G = \E*{X|Y}$ is approximated by an inner Monte Carlo estimate. Further experiments are given for digital option pricing, involving an approximation of a $d$-dimensional SDE.
\end{abstract}

\section{Introduction}
This paper proposes general, efficient numerical methods to compute 
\begin{equation}\label{eqn:multidim-prob}
\prob{G\in\Omega}=\E*{\I{G\in\Omega}},\qquad 
\I{G\in \Omega} \defeq 
\begin{cases}
1 & G\in \Omega \\
0 & G \not \in \Omega 
\end{cases},
\end{equation}
where $G$ is a $d$-dimensional random variable which cannot be sampled directly. In \Cref{sect:assumpt}, we relate \eqref{eqn:multidim-prob} to the one-dimensional problem
\begin{equation}
\label{eqn:heaviside-problem}
	\prob{g>0} = \E{\H{g}},
\end{equation}
where $\H{g}$ is the Heaviside function, equal to 1 when $g\ge 0$ and to 0 otherwise. In most problems of interest, $g$ requires approximate sampling. We assume access to a hierarchy of increasingly accurate approximations $\{g_\ell\}_{\ell \in \mathbb{N}}$ converging to $g$ almost surely as $\ell\to \infty$. Approximate simulation of $g$ induces a bias in typical Monte Carlo methods for \eqref{eqn:heaviside-problem}, increasing the cost of standard Monte Carlo averages. In such situations, Multilevel Monte Carlo (MLMC) \cite{giles:2015review, Giles2008MLMC, Cliffe:2011} is often able to reduce the cost, but is known to suffer when the observable is discontinuous as in \eqref{eqn:multidim-prob} or \eqref{eqn:heaviside-problem} \cite{GilesMikeIMMC:2006, GilesHajiAli:2018, Elfverson:2016selectiverefinement}. Adaptive sampling techniques \cite{GilesHajiAli:2018,Brodie:2011,Elfverson:2016selectiverefinement} have proven successful in reducing the cost of Monte Carlo and MLMC for specific instances of \eqref{eqn:heaviside-problem}. This paper builds upon such methods to establish a general framework for this problem with an emphasis on ensuring applicability to wide ranging problems. Examples are discussed below. 
\begin{example}[Nested Simulation]
	\label{ex:nest}
	Equation \eqref{eqn:heaviside-problem} often arises in financial risk estimation. For example, many risk measures involve nested expectations of the form $g = \E*{X|Y}$ for some random variables $X, Y$ \cite{Gordy:2010, GilesHajiAli:2018, GilesHajiAli:2019sampling, GregoryJon2012Ccra}. Approximation of $g$ by $g_\ell$ is possible using an inner Monte Carlo average with $N_\ell\in \mathbb{N}$ samples. 
\end{example}
\begin{example}[Digital Option Pricing]
	\label{ex:sde}
	Let $S$ be the solution to the d-dimensional SDE
	\[
	\text{d}S(t) = a(t, S(t))\text{d}t + b(t, S(t))\text{d}W(t),
	\]
	at maturity $T$. If $S$ denotes the price of certain assets at maturity $T$, we set $G\equiv S$ and consider a financial derivative with unit payoff if $G\in \Omega$ and no payoff otherwise. The (non-discounted) value at time 0 of this option is of the form \eqref{eqn:multidim-prob}, where $G$ can be approximately sampled using SDE discretisation methods \cite{Kloeden:1999}. 
\end{example}

A related setup is discussed in \cite{Elfverson:2016selectiverefinement} and applied in \cite{Dodwell:2021refinement} to compute failure properties of systems governed by PDEs. In \cite{Elfverson:2016selectiverefinement}, the idea of selective refinement is used to adaptively refine MLMC samples based on the uncertainty of $g>0$. Selective refinement aims to reduce the cost of sampling level $\ell$ without affecting the approximation error of $\H{g}$. There, it is assumed that the error $\abs{g-g_\ell}$ is bounded when $g_\ell$ is near zero, excluding applications like \Cref{ex:nest,ex:sde}.

There is extensive research into Monte Carlo approximation of nested simulation problems as in \Cref{ex:nest}. Analysis of standard Monte Carlo methods for nested simulation is discussed in \cite{Gordy:2010}. Adaptivity is then combined with standard Monte Carlo methods for this problem in \cite{Brodie:2011}. Moreover, in \cite{GilesHajiAli:2018, GilesHajiAli:2019sampling} adaptive MLMC methods for nested simulation are discussed. Contrary to the selective refinement algorithm in \cite{Elfverson:2016selectiverefinement}, these methods aim to improve the approximation error of $\H{g}$ at level $\ell$ while keeping the work of sampling at level $\ell$ unaffected. This approach forms the basis for the present work.

An alternative approach to compute \eqref{eqn:heaviside-problem} via MLMC is to approximate $\H{g}$ by a Lipschitz function. This smoothing procedure is discussed in \cite{Giles:2015Smoothing} and an alternative smoothing procedure for a class of SDEs is discussed in \cite{bayerchiheb:2020}. These approaches require an explicit smoothing step, which the work presented here removes by using adaptivity to implicitly smooth the problem within the MLMC computation itself.\\

The key contributions of this paper are as follows:
\begin{itemize}
	\item A generalisation of the adaptive MLMC sampling scheme for nested simulation \cite{GilesHajiAli:2018, GilesHajiAli:2019sampling} is presented in \Cref{alg:ad-samp}. The new procedure requires less restrictive moment bounds on $g$ and is formulated in a general framework allowing for applications beyond nested simulation. 
	\item Numerical experiments show the adaptive MLMC scheme introduced here remains effective for nested simulation, with a slight relaxation of the sampling process used in \cite{GilesHajiAli:2018,GilesHajiAli:2019sampling}. Additional results show the scheme has an equally strong impact when applied to digital option pricing as in \Cref{ex:sde}.
\end{itemize}

\Cref{sect:assumpt} outlines the problem setup and necessary assumptions for this analysis, before discussing the link between problems \eqref{eqn:multidim-prob} and \eqref{eqn:heaviside-problem}. We describe the MLMC approach to \eqref{eqn:heaviside-problem} in \Cref{sect:mlmc-probs} and show how the complexity of MLMC suffers because $\H{g}$ is discontinuous. In \Cref{sect:ad-mlmc}, we introduce the adaptive MLMC procedure and analyse its benefits to the MLMC complexity. Numerical results are then presented in \Cref{sect:num-ex}. 

 \subsection{Problem Setup}
\label{sect:assumpt}
For the majority of this paper, we focus on the problem \eqref{eqn:heaviside-problem}. At the end of this section, we discuss how to extend the methods to general problems of the form \eqref{eqn:multidim-prob}. As is typical for MLMC, we assume the expected sampling cost of $g_\ell$, denoted $W_\ell$, increases geometrically with $\ell$. In particular,
\begin{equation}
\label{eqn:Wl}
W_\ell \lesssim 2^{\gamma\ell},\quad \text{for some } \gamma >0.
\end{equation}
The following assumption controls the strong approximation error of $g_\ell$.
\begin{assumption}
	\label{assumpt:q-moments}
	For some $2<q$, $\beta>0$ and positive valued random variable $\sigma_\ell$, define
	\begin{equation}
	\label{eqn:Z_ell}
	Z_\ell \defeq \frac{g_\ell-g}{\sigma_\ell 2^{-\beta\ell/2}},
	\end{equation}
	and assume $\E*{\abs*{Z_\ell}^q}$ is uniformly bounded in $\ell\ge 0$.
\end{assumption}
In this context, $\sigma_\ell$ represents fluctuations in the approximation uncertainty for a given instance of $g_\ell$. To implement MLMC successfully, we control the probability of sampling $g_\ell$ close to 0. In doing so, we introduce the parameter
\begin{equation}\label{eqn:delta}
	\delta_\ell \defeq \frac{g_\ell}{\sigma_\ell},
\end{equation}
which models the sample specific uncertainty in the sign of $g_\ell$ and thus $\H{g_\ell}$.
\begin{assumption}
		\label{assumpt:delta-bound}
		There exists $\delta, \rho_0> 0$ such that for all $x \le \delta$ we have
		\[
			\prob{\abs{\delta_\ell}  < x} \le \rho_0 x
		\]
		for all $\ell\ge 0$.
\end{assumption}

\Cref{assumpt:q-moments,assumpt:delta-bound} are enough to bound the strong error of approximations $\H{g_\ell}$, which underpins the complexity theory for MLMC approximation of \eqref{eqn:heaviside-problem}. In certain cases, tight bounds on the cost of MLMC require tight bounds on $\abs*{\E*{\H{g} - \H{g_\ell}}}$, which requires further assumptions (see \Cref{sect:mlmc-probs}).\\

It is important to remark here that the assumptions above allow for the simple extension to the general problem \eqref{eqn:multidim-prob} under equivalent assumptions. To see this, assume that (for $\norm{\cdot}$ being the Euclidean norm)
\[
	\dom{G}\defeq \min_{\omega\in \partial\Omega}\{\left \|G -\omega \right \| \}
\]
exists. Here, we are assuming the minimum distance to the boundary of $\Omega$ is attained by a point on the boundary. Then, \eqref{eqn:multidim-prob} is equivalent to \eqref{eqn:heaviside-problem} when
\[
	g = \dbar{G} \defeq
	\begin{cases}
		\dom{G} & G \in \Omega \\
		-\dom{G} & G \not\in \Omega
	\end{cases}
\]
is a signed distance. If we denote approximations of $G$ at level $\ell\in \mathbb{N}$ by $G_\ell$ then we have approximations $g_\ell \defeq \dbar{G_\ell}$ of $g$. The following result shows \Cref{assumpt:q-moments} holds under a similar condition on $G$.
\begin{lemma}
	\label{lemma:multidim}
	Assume that for some $2<q<\infty$ and random variable $\sigma_\ell >0$
	\[
		\E*{\p*{\frac{\left\|G - G_\ell \right\|}{\sigma_\ell2^{-\beta\ell/2}}}^q}<\infty
	\]
	holds for all $\ell \ge 0$. Then, \Cref{assumpt:q-moments} holds for $g\defeq \dbar{G}, g_\ell \defeq \dbar{G_\ell}$.
\end{lemma}
\begin{proof}
	It is enough to show $\abs{g-g_\ell} \le \left\|G- G_\ell\right\|$ almost surely. There are two possibilities:

	The first option is either $G, G_\ell\in \Omega$ or $G, G_\ell\in \Omega^c$.  By symmetry, we can assume $G, G_\ell \in \Omega$ and $g \ge g_\ell$. Then there is $\omega, \omega^\prime$ such that $g = \left\|G - \omega \right\|$ and $g_\ell = \left\|G_\ell - \omega^\prime \right\|$ and
	\begin{align*}
	\abs{g-g_\ell} &= \left\| G - \omega \right \| - \left\|G_\ell -\omega^\prime\right\| \\
	&\le \left\| G - \omega^\prime \right \| - \left\|G_\ell -\omega^\prime\right\|  \\
	&\le \left\|G - G_\ell \right\|,
	\end{align*}
	by the reverse triangle inequality.

	Alternatively, we may assume $G\in \Omega$ and $G_\ell\in \Omega^c$. Then, there is $\hat\omega\in \partial \Omega$ on the straight line between $G, G_\ell$. For some $\omega, \omega^\prime$, we have
	\begin{align*}
	\abs{g-g_\ell} &= \left\| G - \omega \right \| + \left\|G_\ell -\omega^\prime\right\| \\
	&\le \left\| G - \hat \omega \right \| + \left\|G_\ell -\hat \omega\right\|\\
	&= \left\|G - G_\ell \right\|,
	\end{align*}
	where we use the fact that $\hat\omega$ is on the line from $G$ to $G_\ell$.
\end{proof}
Moreover, \Cref{assumpt:delta-bound} becomes an equivalent condition on the distribution of $\abs{\delta_\ell} = \dom{G_\ell}/\sigma_\ell$, where $\sigma_\ell$ satisfies the assumption of \Cref{lemma:multidim}.
 
\section{Multilevel Monte Carlo for Probabilities}
\label{sect:mlmc-probs}
In this section, we outline the use of standard MLMC methods \cite{Giles2008MLMC,Cliffe:2011,giles:2015review} for approximating \eqref{eqn:heaviside-problem}. In particular, we show that the discontinuity at 0 in the Heaviside function limits the effectiveness of standard MLMC for this problem. We begin by approximating $\prob*{g>0}$ by $\prob{g_L >0}$, where $L$ should be chosen large enough to control the approximation bias. Sampling $g$ at large levels $L$ is typically expensive. The key idea of MLMC  is to split this computation over levels $0\le \ell\le L$ using a telescopic sum. Specifically, using $\H{g_{-1}}\defeq 0$ 
\begin{equation}
\label{eqn:mlmc-estimator}
\begin{aligned}
\E*{\H{g}} \approx \E*{\H{g_L}} &= \sum_{\ell = 0}^L \E*{\dh}\\
&\approx \sum_{\ell = 0}^L\p*{\frac{1}{M_\ell}\sum_{m=1}^{M_\ell} \p*{\dhm}},\\
\end{aligned}
\end{equation}
where we approximate each expectation in the telescopic sum by an independent Monte Carlo sum with samples $\dhm\overset{\text{i.i.d}}{\sim}\dh$. Samples $g_\ell^{(f,m)}$ and $g_{\ell-1}^{(c,m)}$ should be closely correlated to reduce $\var{\dh}$, lowering the number of samples, $M_\ell$ required at level $\ell$. The following result bounds the total work of sampling \eqref{eqn:mlmc-estimator} within a given error tolerance. In the statement of this result and throughout, we use the operator $f_0\lesssim f_1$ to denote $f_0\le C\times f_1$. Here, $f_0, f_1$ depend on the problem parameters, specifically $g, g_\ell,\ell$ and the error bound $\tol^2$ defined in \Cref{thm:complexity}, whereas $0<C<\infty$ is an absolute constant specific to $f_0$ and $f_1$. In particular, $C$ is independent of $\ell$ and the error bound $\tol^2$.    \\

\begin{theorem}[\cite{Cliffe:2011,giles:2015review}]
	\label{thm:complexity}
	Let $\{\dhl\}_{\ell=0}^\infty$ be a sequence of random variables with $\prob{g>0} = \sum_{\ell=0}^\infty \E{\dhl}$. Assume the following rates of convergence for some $\gamma, \beta_\textnormal{ind}>0,\  \alpha_\textnormal{ind} \ge \frac{\min \p{\gamma, \beta_\textnormal{ind}}}{2}$:
		\begin{itemize}
			\item The expected work of sampling $\dhl$ is $W_\ell \lesssim 2^{\gamma\ell}$.
			\item The mean and variance of $\dhl$ converge to 0 with the following rates
			\begin{align}
			E_\ell \defeq \abs*{\E*{\dhl}} &\lesssim 2^{-\alpha_\textnormal{ind}\ell}. \label{weak-err}\\
			V_\ell \defeq \var{\dhl} &\lesssim 2^{-\beta_\textnormal{ind}\ell}, \label{str-err}
			\end{align}
		\end{itemize}
	Then, there is optimal $L$ and $\{M_\ell \}_{0\le\ell\le L}$ such that the total work of computing the MLMC estimator
	\begin{equation}\label{eqn:mlmcl}
	\mathcal{M}_{M_0,\dots,M_L}^{L}\defeq \sum_{\ell=0}^{L}  \p*{\frac{1}{M_\ell}\sum_{m=1}^{M_\ell} \dhl^{(m)} }, \quad \dhl^{(m)}\iid \dhl
	\end{equation}
	with mean square error satisfying 
	\(
		\E{\p{\prob{g>0} - \mathcal{M}_{M_0,\dots,M_L}^{L}\defeq \sum_{\ell=0}^{L} }^2} \le \tol^2
	\)
	is
	\[
	\textnormal{Work}\p*{\mathcal{M}_{M_0,\dots,M_L}^{L}, \tol} \lesssim 
	\begin{cases}
	\tol^{-2} & \beta_\textnormal{ind}>\gamma \\
	\tol^{-2}\p*{\log\tol}^2 & \beta_\textnormal{ind} = \gamma \\
	\tol^{-2-\p*{\gamma - \beta_\textnormal{ind}}/\alpha_\textnormal{ind}} & \beta_{\textnormal{ind}} < \gamma 
	\end{cases}.
	\]
	We will denote the estimator \eqref{eqn:mlmcl} with optimal $L$ and $\{M_\ell \}_{0\le\ell\le L}$ by $\mlmc$.
\end{theorem}
\begin{remark}
	\label{remark:complexity}
	\Cref{thm:complexity} can be applied to the MLMC estimator \eqref{eqn:mlmc-estimator} by taking $\dhl \defeq \dh$, in \Cref{sect:ad-mlmc} we see $\dhl$ take a slightly different form to accommodate adaptive approximation of $g$. $E_\ell$ \eqref{weak-err} and $V_\ell$ \eqref{str-err} are the  bias and variance of the multilevel correction, respectively. Rather than prove convergence rates for these terms directly, we provide stronger results on $\E{\p*{\H{g} - \H{g_\ell}}^2}, \abs{\E{\H{g}-\H{g_\ell}}}$. The bound on $\work$ is sometimes referred to as the complexity of $\mlmc$, since it describe how the total work scales as the error decreases. Replacing $\H{\cdot}$ with a smooth/Lipschitz functional, it follows from \Cref{assumpt:q-moments} and \cite{Cliffe:2011, giles:2015review} that \Cref{thm:complexity} holds for $\beta_{\textnormal{ind}}=\beta$ and we see $\tol^{-2}$ complexity for $\beta > \gamma$.  In this paper, we refer to $\tol^{-2}$ as the `canonical' complexity since it is the same as seen for standard Monte Carlo with exact sampling of $g$.
\end{remark}

The following result provides a bound on $\E{\p{\dh}^2}$ under the assumptions in \Cref{sect:assumpt}. The rate is worse than that of smooth/Lipschitz functionals mentioned in \Cref{remark:complexity}, since we make an $\Order{1}$ approximation error in $\H{g} - \H{g_\ell}$ whenever $g, g_\ell$ lie on opposite sides of 0. 
\begin{proposition}[Variance With General Assumptions]
	\label{thm:standard-sampling-var}
	Under \Cref{assumpt:q-moments,assumpt:delta-bound} we have
	\(
	\E{\p{\H{g} - \H{g_\ell}}^2} \lesssim 2^{-\p{\frac{q}{q+1}}\ell\beta/2}.
	\)
\end{proposition}
\begin{proof}
	We compute
	\[
	\begin{aligned}
	\E{\p{\H{g} - \H{g_\ell}}^2} &\leq \E{\I{\abs{g - g_\ell} \geq \abs{g_\ell}}} \\
	&= \E{\I{\abs{Z_\ell} \geq 2^{\ell\beta/2}\abs{\delta_\ell}}},
	\end{aligned}
	\]
	where $Z_\ell$ and $\delta_\ell$ are as in \eqref{eqn:Z_ell} and \eqref{eqn:delta}. It follows that
	\[
	\begin{aligned}
	\E{\I{\abs{Z_\ell} \geq 2^{\ell\beta/2}\abs{\delta_\ell}}} &= \E{\I{\abs{Z_\ell} \geq 2^{\ell\beta/2}\abs{\delta_\ell}}\I{\abs{\delta_\ell}\leq\psi}} + \E{\I{\abs{Z_\ell} \geq 2^{\ell\beta/2}\abs{\delta_\ell}}\I{\abs{\delta_\ell} \geq \psi}} \\
	&\leq \E{\I{\abs{\delta_\ell}\leq\psi}} + \E{\I{\abs{Z_\ell} \geq 2^{\ell\beta/2}\psi}} \\
	&\leq \rho_0\psi +(2^{\ell\beta/2}\psi)^{-q}\E{\abs{Z_\ell}^q},
	\end{aligned}
	\] 
	where we have used \Cref{assumpt:delta-bound}. Then we set  $\psi = \min(1, \delta)2^{-\p{\frac{q}{q+1}}\ell\beta/2}$ to get the previous two terms of equal rate, which is the variance convergence rate.
\end{proof}
\begin{remark}
	\label{remark:weakbound}
	\Cref{thm:standard-sampling-var} also proves an upper bound on $E_\ell$ for $\dhl=\dh$ \eqref{weak-err} since we have
	\(
		\abs*{\E*{\dh}} \le \E*{\p*{\dh}^2}.
	\)
\end{remark}
In the context of \Cref{thm:complexity}, \Cref{thm:standard-sampling-var} shows $\beta_{\textnormal{ind}} = \p{\frac{q}{q+1}}\frac{\beta}{2}$ and we only observe $\tol^{-2}$ complexity when $\beta>2\p{\frac{q+1}{q}}\gamma$. In many examples, including those discussed here, $\beta \le 2\gamma$ and we need tight bounds on $E_\ell$ \eqref{weak-err} to state accurate complexities. To derive tighter bounds than \Cref{remark:weakbound} we require further assumptions.
\begin{assumption}[\cite{Gordy:2010}]
	\label{assumpt:weak-err-st}
	Let $\rho_\ell (y, z)$ be the joint density of $\delta_\ell$ \eqref{eqn:delta} and $Z_\ell$ \eqref{eqn:Z_ell}, defined for some $\beta>0$. Assume that for all $\ell$, $\rho_\ell$ is twice differentiable in $y$ and
	\[
	\bigg|\frac{\partial^i}{\partial y^i}\rho_\ell(y, z) \bigg| \leq p_{i, \ell}(z), \quad \sup_\ell \int_{\mathbb{R}}|z|^{j} p_{i, \ell}(z)\text{d}z < \infty,
	\]
	for $i = 0, 1, 2$ and   $0\le j\le q+2$ for some $q > 2$.
\end{assumption}

\begin{assumption}
	\label{assumpt:weak-err-nondiscont}
	For $Z_\ell, \beta$ as in \Cref{assumpt:q-moments}, we have
	\(
	\left |\E*{Z_\ell} \right | \lesssim 2^{\ell\p*{\beta/2-\alpha}},
	\)
	for some $\frac{\beta}{2}\le \alpha \le \beta$.
\end{assumption}
From \eqref{eqn:Z_ell} we see that \Cref{assumpt:weak-err-nondiscont} bounds $\E*{\sigma_\ell^{-1}\p*{g-g_\ell}}\lesssim 2^{-\alpha\ell}$. \Cref{assumpt:weak-err-nondiscont} is instead expressed in terms of $Z_\ell$ to align with the analysis in \Cref{sect:ad-mlmc} (see \Cref{assumpt:weak-err-adapt}). We stress that these assumptions are required only to obtain better convergence rates of $E_\ell$. Reasonable results can still be obtained using \Cref{remark:weakbound} when they are false. Nonetheless, \Cref{assumpt:weak-err-st} also provides slightly better bounds for $\E{\p{\dh}^2}$. For completeness, we state this result below. 
\begin{proposition}[Variance With Strict Assumptions]
	\label{prop:str-err-strict-assumptions}
	Under \Cref{assumpt:weak-err-st} it follows that
	\(
	\E*{\p*{\H{g} - \H{g_\ell}}^2} \lesssim 2^{-\ell\beta/2}.
	\)
\end{proposition}
\begin{proof}
	We have 
	\[
	\begin{aligned}
	\E*{\p*{\H{g} - \H{g_\ell}}^2} &\le \E*{\I{\abs{Z_\ell}\ge b_\ell\abs{\delta_\ell}} }\\
	& = \int_{\R}\int_{-b_\ell^{-1}\abs{z}}^{b_\ell^{-1}\abs{z}} \rho_\ell (y, z)\text{d}y\text{d}z \\
	& = \int_{\R}\int_{-b_\ell^{-1}\abs{z}}^{b_\ell^{-1}\abs{z}} \rho_\ell (0, z) + y\frac{\partial}{\partial y}\rho_\ell(0, z) + \frac{y^2}{2}\frac{\partial^2}{\partial y^2}\rho_\ell(\hat y, z)\text{d}y\text{d}z.
	\end{aligned}
	\]
	If $\int_{\R}\rho_\ell(0,z)\text{d}z = 0$, we can eliminate the contribution of $\rho_\ell(0, z)$, thus we may assume $\int_{\R}\rho_\ell(0,z)\text{d}z > 0$. \Cref{assumpt:weak-err-st} shows that the double integral over $\rho_\ell(0, z)$ gives the dominant term. Thus,
	\[
	\begin{aligned}
	\E*{\p*{\H{g} - \H{g_\ell}}^2} &\lesssim b_\ell^{-1} \int_{\R}\abs{z}\rho_\ell(0, z)\text{d}z \\
	&= 2^{-\ell\beta/2} \E*{\abs{Z_\ell}\ \vline \ \delta_\ell = 0}
	\end{aligned}
	\]
	proving the result for $\E*{\abs{Z_\ell}\ \vline \ \delta_\ell = 0} < \infty$.
\end{proof}

The stricter conditions also give a tighter bound on the $E_\ell$ than \Cref{remark:weakbound}, and hence better MLMC complexity when $\beta < 2\gamma$.
\begin{proposition}[\cite{Gordy:2010} Proposition 1]
	\label{thm:standard-bias}
	Let \Cref{assumpt:weak-err-st,assumpt:weak-err-nondiscont} hold for some $\beta>0,\ \frac{\beta}{2}\le\alpha\le \beta$. Then, 
	\(
	\left |\E{\H{g} - \H{g_\ell}}\right | \lesssim 2^{-\alpha\ell}.
	\)
\end{proposition}
\begin{proof}
	For $\rho_\ell(y,z)$ given by \Cref{assumpt:weak-err-st} we have
	\[
	\E{\H{g}} =  \int_\mathbb{R} \int_{2^{-\beta\ell/2}z}^\infty \rho_\ell(y, z)\text{d}y\text{d}z.
	\]
	Thus
	\[
	\begin{aligned}
	\left |\E{\H{g} - \H{g_\ell}}\right | &= \left|\E{\H{g_\ell}} - \E{\H{g}}\right| \\
	&= \left|\int_\mathbb{R} \int_{0}^{2^{-\beta\ell/2}} \rho_\ell(y, z)\text{d}y\text{d}z\right|.
	\end{aligned}
	\]
	A Taylor expansion gives
	\begin{equation}
	\label{eqn:taylor-rho}
	\rho_\ell(y, z) = \rho_\ell(0, z) + y\frac{\partial}{\partial y}\rho_\ell(0, z) + \frac{y^2}{2}\frac{\partial^2}{\partial y^2}\rho_\ell(\hat y, z),
	\end{equation}
	for some $\hat y\in [0, y]$. Inserting this into the double integral above and using \Cref{assumpt:weak-err-st,assumpt:weak-err-nondiscont} gives 
	\[
	\begin{aligned}
	E_\ell &\leq \left|2^{-\beta\ell/2}\int_\mathbb{R} z\rho_{\ell}(0, z)\text{d}z\right| + 2^{-\beta\ell}\int_\mathbb{R} |z|^2 p_{1, \ell}(z)\text{d}z \\
	&\qquad + 2^{-3\beta\ell/2}\int_\mathbb{R}|z|^3p_{2, \ell}(z)\text{d}z\\
	&\lesssim 2^{-\beta\ell/2}\left|\E{Z_\ell\ \vline\ \delta_\ell = 0}\right| + \Order*{2^{-\beta\ell}} \\
	&\lesssim 2^{-\alpha\ell},
	\end{aligned}
	\]
	where we used \Cref{assumpt:weak-err-nondiscont} and the definition of $Z_\ell$ to bound $\E{Z_\ell\ \vline\ \delta_\ell = 0} \lesssim 2^{\ell\p*{\beta/2-\alpha}}$ and assume $\int_{\R}\rho_\ell(0,z)\text{d}z > 0$ as in the proof of \Cref{prop:str-err-strict-assumptions}.
\end{proof}

The discussion above proves the following complexity results.
\begin{corollary}
	\label{cor:st-comp-weak}
	Under \Cref{assumpt:q-moments,assumpt:delta-bound}, the total work required for the MLMC estimator \eqref{eqn:mlmc-estimator} with mean square error $\tol^2$ can be bounded by
	\[\work\lesssim
	\begin{cases}
	\tol^{-2} & \beta > 2\p{\frac{q+1}{q}}\gamma \\
	\tol^{-2}\p*{\log\tol}^2& \beta = 2\p{\frac{q+1}{q}}\gamma \\
	\tol^{-1-2\p{\frac{q+1}{q}}\gamma/\beta} & \beta < 2\p{\frac{q+1}{q}}\gamma 
	\end{cases}.
	\] 
\end{corollary}
\begin{proof}
	The result follows by combining \Cref{thm:standard-sampling-var} and \Cref{remark:weakbound} with \Cref{thm:complexity} for $\dhl=\dh$.
\end{proof}

\begin{corollary}
	\label{cor:st-comp-str}
	Under \Cref{assumpt:weak-err-st} and, when $\beta<2\gamma$, also under \Cref{assumpt:weak-err-nondiscont} the total work required for the MLMC estimator \eqref{eqn:mlmc-estimator} with mean square error $\tol^2$ can be bounded by
	\[
	\work \lesssim
	\begin{cases}
	\varepsilon^{-2} & \beta > 2\gamma\\
	\varepsilon^{-2}\p*{\log\varepsilon}^2 & \beta = 2\gamma\\
	\varepsilon^{-2 - \p*{\gamma-\beta/2}/\alpha} & \beta < 2\gamma 
	\end{cases}
	\]
\end{corollary} 
\begin{proof}
	The result follows by combining  \Cref{thm:standard-sampling-var} and \Cref{thm:standard-bias} with \Cref{thm:complexity} for $\dhl=\dh$.
\end{proof}
In some applications, \Cref{assumpt:q-moments} holds for all $q<\infty$. In this case, \Cref{cor:st-comp-weak} holds by taking $q\to \infty$, where one must multiply the complexity by a factor $\tol^{-\nu}$ for any $\nu>0$ when $\beta\le 2\gamma$ as a technical restraint. For \Cref{ex:nest,ex:sde} with Euler-Maruyama simulation of the SDE, we can show (under certain assumptions on the underlying SDE \cite{Kloeden:1999}) that $\alpha=\beta=\gamma$  the complexity is at best $\tol^{-2.5}$, a significant increase over the canonical $\tol^{-2}$ complexity. For SDE simulation we can replace the $\tol^{-\nu}$ term appearing in the complexity  in the limit $q\to\infty$ with a logarithmic factor using the analysis in \cite{Avikainen:irregularSdes}.
 
\section{Adaptive Multilevel Monte Carlo}
\label{sect:ad-mlmc}
In the previous section, we described how the complexity of MLMC calculations for the problem \eqref{eqn:heaviside-problem} is affected by the discontinuous observable $\H{g}$. To improve the performance of MLMC we replace the approximation $g_\ell$ at level $\ell$ with $g_{\ell+\eta_\ell}$. Where we introduce the random, non-negative, integer $\eta_\ell$ which should reflect the uncertainty in the sign of $g_{\ell+\eta_\ell}$. The MLMC estimator \eqref{eqn:mlmcl} then uses the multilevel correction term $\dhl$ given by
\begin{equation}
\label{eqn:dh-ad}
\dhl \defeq 
\begin{cases}
\H{g_{\ell+\eta_\ell}} - \H{g_{\ell-1+\eta_{\ell-1}}} & \ell > 0\\
\H{g_{\eta_0}} & \ell = 0
\end{cases}.
\end{equation}
Heuristically, approximations which are close to zero with high variability should be refined further (have larger values of $\eta_\ell$) than approximations which lie far away from zero with low variability. The chosen approach for sampling $g_{\ell+\eta_\ell}$ is detailed in \Cref{alg:ad-samp}. We refine between levels $\ell\le \ell+\eta_\ell \le \ell + \ceil{\theta\ell}$, for a supplied parameter $\theta$, based on the value of $\abs{\delta_{\ell+\eta_\ell}}$ \eqref{eqn:delta}. \Cref{alg:ad-samp} also has the parameter $r$, determining how strict we are with the refinement, and a confidence constant $c>0$. Explicitly, we refine by $\eta_\ell$ levels, where 
\begin{equation}
\label{eqn:xi-ell-cond}
\eta_\ell =  k \iff
\begin{cases}
\abs{\delta_{\ell + m}} < c2^{\gamma(\theta\ell(1-r)-m)/r} & \forall m\leq k-1\\
\abs{\delta_{\ell + k}} \geq c2^{\gamma(\theta\ell(1-r) - k)/r} & \text{if } k<\theta\ell
\end{cases},
\end{equation}
for $0\le k\le \ceil{\theta\ell}$. For small values of $r$ we refine samples to higher levels than for large $r$. Ideally, we want to allow the refinement procedure to take $r$ as large as possible while observing maximum benefit to the MLMC complexity. Within \Cref{alg:ad-samp}, it is important that the method of refining $g_{\ell+\eta_\ell}$ to $g_{\ell+\eta_\ell+1}$ does not affect the almost sure convergence of $g_{\ell+\eta_\ell}$ to $g$.\\

\begin{center}
	\begin{algorithm}
		\KwIn{$\ell, r, \theta, c>0, \gamma, \beta$}
		\KwResult{Adaptively refined sample $g_{\ell+\eta_\ell}$}
		Set $\eta_\ell = 0$\;
		Sample $g_\ell$\;
		Compute $\delta_\ell$ given $g_\ell$\;
		\While{$\abs{\delta_{\ell+\eta_\ell}} < c2^{\gamma\p{\theta\ell\p{1-r}-\eta_\ell}/r}$ and $\eta_\ell < \ceil{\theta\ell}$}{
			Refine $g_{\ell+\eta_\ell}$ to $g_{\ell+\eta_\ell + 1}$\;
			Compute $\delta_{\ell+\eta_{\ell}+1}$ given $g_{\ell+\eta_{\ell}+1}$\;
			Set $\eta_\ell = \eta_\ell+1$\;
		}
		\KwOut{$g_{\ell+\eta_\ell}$}
		\caption{Adaptive sampling at level $\ell$}
		\label{alg:ad-samp}
	\end{algorithm}
\end{center}

\Cref{alg:ad-samp} has many similarities to the adaptive nested simulation algorithm in \cite{GilesHajiAli:2018, GilesHajiAli:2019sampling}, which considers the specific case $g = \E*{X|Y}$ approximated by an inner Monte Carlo sampler.  However, besides being applicable to a wider class of problems, the present algorithm has some key differences: The nested simulation algorithm in \cite{GilesHajiAli:2018,GilesHajiAli:2019sampling} requires that each refined value $g_{\ell+\eta_\ell+1}$ is independent of the previous term $g_{\ell+\eta_\ell}$ conditioned on $Y$, which is not required here. This accelerates the refinement procedure since one can reuse all terms from the computation of $g_{\ell+\eta_\ell}$ in the refinement to $g_{\ell+\eta_\ell+1}$. Moreover, in \cite{GilesHajiAli:2018} the adaptive algorithm returns only the number of inner samples one should use to approximate $\E{X|Y}$, given $Y$, and the estimate of $g$ should then be computed independently. In contrast, our algorithm requires that the estimate of $g$ matches the output of the refinement process.  The parameter $\theta$ is also a novel introduction to \Cref{alg:ad-samp}. In \cite{GilesHajiAli:2018}, the work has $\beta=\gamma$ and under the assumptions there it makes sense to choose $\theta = 1$ (see \Cref{thm:var-bound-strong}). For $\beta \neq \gamma$ it can be optimal to refine over a wider or narrower range of levels, see \Cref{thm:var-bound} and \Cref{remark:large-theta}.\\

In certain applications, it is possible that samples at the fine and coarse levels within MLMC are correlated to such extent that when $\eta_{\ell-1} = \eta_{\ell}+1$ we have $g_{\ell + \eta_\ell} = g_{\ell-1 + \eta_{\ell-1}}$. In this case, when $r\le2$ it follows from \eqref{eqn:xi-ell-cond} that $\ell - 1 + \eta_{\ell-1} \le \ell + \eta_\ell$. However, when $r>2$ there is a small chance that this is false and the `coarse' sample $g_{\ell-1 + \eta_{\ell-1}}$ is actually refined to greater accuracy than the `fine' estimator $g_{\ell + \eta_{\ell}}$. Here, we can resort to \Cref{prop:work-rate} below which assures that on average $g_{\ell + \eta_{\ell}}$ has greater accuracy than $g_{\ell-1 + \eta_{\ell-1}}$.

\subsection{Work Analysis}
\label{sect:work-analysis}
In the context of \Cref{thm:complexity}, using $\dhl$ as in \eqref{eqn:dh-ad} we wish to improve upon the convergence rate of $V_\ell$ seen for the estimator \eqref{eqn:mlmc-estimator} in \Cref{thm:standard-sampling-var}.  \Cref{thm:complexity} implies that for this to be effective the expected cost of computing $g_{\ell+\eta_\ell}$ and $g_\ell$ must be similar. The following result ensures the expected cost of sampling $g_{\ell + \eta_\ell}$ is also $\Order{2^{\gamma\ell}}$.

\begin{proposition}[\cite{GilesHajiAli:2018} Theorem 2.7]
	\label{prop:work-rate}
	Define $\eta_\ell$ as in \eqref{eqn:xi-ell-cond} and assume \Cref{assumpt:delta-bound} holds for fixed $\rho_0, \delta>0$. Provided $r>1$, we have
	\[
		\E*{2^{\gamma\p*{\ell + \eta_\ell}}} \lesssim 2^{\gamma\ell}.
	\]
\end{proposition}
\begin{proof}
We start with 
\[
\begin{aligned}
\E{2^{\gamma\p{\ell + \eta_\ell}}} &= \sum_{k=0}^{\left\lceil \theta \ell \right \rceil} 2^{\gamma(\ell + k)}\prob{\eta_\ell = k}\\
&\leq 2^{\gamma\ell}  + \sum_{k=1}^{\left\lceil \theta \ell \right \rceil} 2^{\gamma(\ell+ k)}\prob{\abs{\delta_{\ell+k-1}} < c2^{\gamma(\theta\ell(1-r) - k + 1)/r}},
\end{aligned}
\]
where we have used \eqref{eqn:xi-ell-cond} to bound the probabilities. Provided $r>1$, for large enough $\ell$ we have $ c2^{\gamma(\theta\ell(1-r) - k + 1)/r} < \delta$ 
for all $k \geq 0$. Using \Cref{assumpt:delta-bound}
\[
\begin{aligned}
\E{2^{\gamma\p{\ell + \eta_\ell}}} &\leq 2^{\gamma\ell} + \rho_0c2^{\gamma/r} 2^{\gamma\ell\p*{1 + \theta\p*{1-r }/r }}\sum_{k=1}^{\left\lceil \theta \ell \right \rceil} 2^{\gamma k(r-1)/r} \\
&\leq 2^{\gamma\ell} + c_02^{\gamma\ell} 2^{\gamma\ell\theta\p*{1-r }/r } 2^{\gamma \left\lceil \theta \ell \right \rceil(r-1)/r} \\
&\lesssim2^{\gamma\ell},
\end{aligned}
\]
since $r>1$.
\end{proof}
Note that the above proof emphasises that taking $r$ larger results in a sampling cost that is lower by a constant factor.

\subsection{Analysis of the Variance}

The following results highlight improvements to the convergence of $\E{\p{\dh}^2}$ to 0 under the adaptive sampling procedure in \Cref{alg:ad-samp}. As with the non-adaptive case, we obtain slightly better results using the stronger \Cref{assumpt:weak-err-st}. However, this condition is not essential and we still see an improvement under the general \Cref{assumpt:q-moments,assumpt:delta-bound}, as seen below.
\begin{lemma}\label{thm:var-bound}
   Let \Cref{assumpt:q-moments,assumpt:delta-bound} hold for some $\beta>0$ and $q>2$. Assume:
   \begin{itemize}
   	\item For $\beta \le \p{\frac{q+1}{q}}\gamma$ we take $r < 2\frac{\gamma}{\beta}$ and
   	\begin{equation}
   	\theta = \p*{2\p*{\frac{q+1}{q}}\frac{\gamma}{\beta} - 1}^{-1}.\label{eqn:theta-opt}
   	\end{equation}
   	\item For $\beta> \p{\frac{q+1}{q}}\gamma$ we take $\theta = 1$ and
   	\begin{equation}\label{eqn:rvals-weak}
   	\begin{cases}
   	r \leq \p*{1-\p*{\frac{q-1}{2(q+1)}}\frac{\beta}{\gamma}}^{-1} & \beta < 2\p*{\frac{q+1}{q-1}}\gamma \\
   	r < \infty & \beta \ge 2\p*{\frac{q+1}{q-1}}\gamma
   	\end{cases}.
   	\end{equation}
   \end{itemize}
	Then, for $g_{\ell+\eta_\ell}$ given by \Cref{alg:ad-samp},
	\begin{equation}
		\label{eqn:adaptive-var}
		\E{\p{\H{g}-\H{g_{\ell+\eta_\ell}}}^2} \lesssim 2^{-\p{\frac{q}{q+1}}\ell\beta(1+\theta)/2}.
	\end{equation}
	
\end{lemma}
\begin{proof}

	As with the work analysis, we split the calculation among each value of $\eta_\ell$
	\begin{equation}
	\label{eqn:var-split}
	\begin{aligned}
	\E{\p{\H{g} -\H{g_{\ell+\eta_\ell}}}^{2}} =&  \sum_{k=0}^{\left\lceil \theta\ell \right\rceil} \E{
		\p{\H{g} -\H{g_{\ell + \eta_\ell}}}^{2} \I{{\eta_{\ell} = k}}}\\
	&\leq \underbrace{\sum_{k=0}^{\left\lceil \theta\ell \right\rceil-1} \E{\p{\H{g} -\H{g_{\ell+k}}}^{2} \I{{\eta_{\ell} = k}}}}_{\eqdef\Sigma_0}+ \underbrace{\E{\p{\H{g} -\H{g_{\ell+\left\lceil \theta\ell \right\rceil}}}^{2}}}_{\eqdef\Sigma_1}
	\end{aligned}
	\end{equation}
	 By \Cref{thm:standard-sampling-var} we
        have
	\[
	\Sigma_1 \lesssim 2^{-\p{\frac{q}{q+1}}\ell\beta\p*{1 + \theta}/2 }.
	\]

	We now turn our attention to terms for which $k<\left\lceil \theta\ell \right\rceil$. Using \eqref{eqn:xi-ell-cond} to relate the condition $\eta_\ell = k$ to the value of $\delta_{\ell + k}$  we have
	\[
	\begin{aligned}
		\Sigma_0 &\leq \sum_{k=0}^{\left\lceil \theta\ell \right\rceil-1}\E{\I{\abs{g-g_{\ell+k}}\ge \abs{g_{\ell+k}}} \I{{ \eta_\ell = k}}} \\
	&\leq
	\sum_{k=0}^{\left\lceil \theta\ell \right\rceil-1}\E*{\I{\abs{g-g_{\ell+k}}\ge \abs{g_{\ell+k}}} \ \ \I{\abs{\delta_{\ell+k}} \geq c2^{\p{\theta\ell\p{1-r}-k}\gamma/r}}}\\
	&= \sum_{k=0}^{\left\lceil \theta\ell \right\rceil-1}\E*{\I{
			c2^{\p{\ell+k}\beta/2 + \p{\theta\ell\p{1-r}-k}\gamma/r} \leq \abs{\delta_{\ell+k}} 2^{\p{\ell+k}\beta/2} \leq \frac{\abs{g-g_{\ell+k}}}{\sigma_{\ell+k} 2^{-\p{\ell+k}\beta/2}}}} \\
	&= \sum_{k=0}^{\left\lceil \theta\ell \right\rceil-1}\E*{\I{ a_k \leq\,  \abs{\delta_{\ell+k}} \leq\, \abs{Z_{\ell+k}}b_k^{-1}}}
	\end{aligned}
	\]
	where $Z_{\ell+k}$ is as in \eqref{eqn:Z_ell} and we introduce the terms
	\[
	\begin{aligned}
	a_k &\defeq c\, 2^{\p{\theta\ell\p{1-r}-k}\gamma/r} \\
	b_k &\defeq 2^{\p{\ell+k}\beta/2}
	\end{aligned}
	\]
	Note that,
	\[
	\begin{aligned}
	\I{ a_k \leq\, \abs{\delta_{\ell+k}} \leq\, \abs{Z_{\ell+k}}b_k^{-1}}&= \I{1 \le a_k^{-q}\abs{\delta_{\ell + k}}^q}\I{\abs{\delta_{\ell + k}}\le \abs{Z_{\ell+k}}b_k^{-1}} \\
	&\le \min\p*{1, a_k^{-q}\abs{\delta_{\ell + k}}^q}\I{\abs{\delta_{\ell + k}}\le \abs{Z_{\ell+k}}b_k^{-1}} \\
	&\le a_k^{-q}b_k^{-q}\abs{Z_{\ell+k}}^q.
	\end{aligned}
	\]
	Therefore, using \Cref{assumpt:q-moments} to bound $\E{\abs{Z_{\ell+k}}^q}$ we obtain
	\[
	\begin{aligned}
	\Sigma_0 &\le \sum_{k=0}^{\left\lceil \theta\ell \right\rceil-1} a_k^{-q}b_k^{-q}\E{\abs{Z_{\ell+k}}^q}.
	\end{aligned}
	\]

	Thus we restrict our attention to the term
	\begin{equation}
	\label{eqn:bk-ak}
	a_k^{-q}b_k^{-q} = 2^{-q(\beta/2 - \gamma/r)(k + \ell)}2^{-q\gamma\ell(1 + \theta -\theta r)/r}.
	\end{equation}
	
	Suppose first that $\beta \le \p{\frac{q+1}{q}}\gamma$. Assume that $r < 2\frac{\gamma}{\beta}$ so that \eqref{eqn:bk-ak} is an increasing function of $k$. It follows that 
	\begin{equation}
		\label{eqn:sum-bk-ak}
	\begin{aligned}
		\sum_{k=0}^{\left\lceil \theta\ell \right\rceil-1} a_{k}^{-q}b_k^{-q}\E{\abs{Z_{\ell+k}}^q} &\lesssim a_{\theta\ell}^{-q}b_{\theta\ell}^{-q} \\
		&\lesssim 2^{q\ell\p*{\gamma\theta - \beta\p*{1+\theta}/2 }}. 
	\end{aligned}
	\end{equation}
	In order to ensure the above term is of the same order as $\Sigma_1$ we take $\theta$ as in \eqref{eqn:theta-opt}.
	
	Now suppose $\beta > \p{\frac{q+1}{q}}\gamma$ and consider first $r < 2\frac{\gamma}{\beta}$ so that \eqref{eqn:sum-bk-ak} holds.  Note that taking $\theta = 1$ is enough to guarantee $\Sigma_1  \lesssim2^{-\p{\frac{q}{q+1}}\beta\ell}$ and, by \eqref{eqn:sum-bk-ak}, $\Sigma_0 \lesssim 2^{q\ell(\gamma - \beta)} \leq2^{-\p{\frac{q}{q+1}}\beta\ell}$ since $\beta \ge \p{\frac{q+1}{q}}\gamma$.  Since $\beta > \p{\frac{q+1}{q}}\gamma$ this is enough to guarantee $\tol^{-2}$ complexity. If $r = 2\frac{\gamma}{\beta}$ the bound \eqref{eqn:sum-bk-ak} becomes (again taking $\theta = 1$)
	\[
	\begin{aligned}
		\Sigma_0\le\sum_{k=0}^{\ell-1} a_{k}^{-q}b_k^{-q}\E{\abs{Z_{\ell+k}}^q} &= \sum_{k=0}^{\ell-1}2^{q\ell(\gamma - \beta)}  \\
		&\lesssim \ell2^{q\ell(\gamma - \beta)}  \lesssim2^{-\p{\frac{q}{q+1}}\beta\ell }.
	\end{aligned}
	\]
	On the other hand, for $r > \frac{2\gamma}{\beta}$, \eqref{eqn:bk-ak} is a decreasing function of $k$ and we have 
	\[
	\begin{aligned}
	\Sigma_0 \le \sum_{k=0}^{\ell - 1} a_{k}^{-q}b_k^{-q}\E{\abs{Z_{\ell+k}}^q} &\lesssim a_{0}^{-q}b_{0}^{-q} \\
		&\lesssim 2^{q\p*{\gamma\p{r-1}/r - \beta/2 }\ell} \lesssim 2^{-\p{\frac{q}{q+1}}\beta\ell },
	\end{aligned}
	\]
	provided we take $r$ as in \eqref{eqn:rvals-weak}, completing the proof.
\end{proof}
\begin{remark}
	\label{remark:large-theta}
		The proof of \Cref{thm:var-bound} allows \eqref{eqn:adaptive-var} to hold for certain values $\theta>1$ provided $\beta> \p{\frac{q+1}{q}}\gamma$ and under tighter upper bounds for $r$. However, for such values of $\beta$ we are already in the $\tol^{-2}$ complexity regime of MLMC at $\theta = 1$, thus any increase in $\theta$ can improve the MLMC cost by a constant at best. Moreover,  tighter bounds on $r$ will increase the expected cost of sampling $g_{\ell+\eta_\ell}$, limiting the value of any constant reduction in the MLMC cost.
\end{remark}
Below, we state an extension to \Cref{thm:var-bound} under the stricter assumptions required for the bias analysis.
\begin{lemma}
	\label{thm:var-bound-strong}
	 Let \Cref{assumpt:weak-err-st} hold for some $\beta>0$ and $q>2$. Assume:
	 \begin{itemize}
	 	\item For $\beta \le \gamma$ we take 
	 	\begin{align}
	 	r &< \frac{2\gamma}{\beta}\p*{1-\frac{1}{q}},\nonumber\\
	 	\theta &= \p*{2\frac{\gamma}{\beta} - 1}^{-1}.\label{eqn:theta-strong}
	 	\end{align}
	 	\item For $\beta>\gamma$ we take $\theta = 1$ and
	 	\begin{equation}\label{eqn:rvals-strict}
	 	\begin{cases}
	 	r \leq \p*{1-\frac{\p*{q-2}\beta}{2(q-1)\gamma}}^{-1} & \beta < 2\frac{q-1}{q-2}\gamma \\
	 	r < \infty & \beta \ge 2\frac{q-1}{q-2}\gamma
	 	\end{cases}.
	 	\end{equation}
	 \end{itemize}
 	
 	Then, for $g_{\ell+\eta_\ell}$ as in \Cref{alg:ad-samp}
	\begin{equation}
	\E{\p{\H{g}-\H{g_{\ell+\eta_\ell}}}^2} \lesssim 2^{-\beta\ell\cdot (1+\theta)/2}.
	\end{equation}
\end{lemma}
\begin{proof}
	As in the previous result, we split the calculation across all refined levels as in \eqref{eqn:var-split}. By \Cref{prop:str-err-strict-assumptions}, it follows that $\Sigma_1 \lesssim 2^{-\beta(1+\theta)\ell/2}$. Moreover, for $k<\ceil{\theta\ell}$ and defining $a_k, b_k$ as in the proof of \Cref{thm:var-bound} we have 
	\[
	\begin{aligned}
		&\E*{\p*{\H{g} - \H{g_{\ell+k}}}^2\I{\eta_\ell = k}} \\
		&= \E*{\p*{\I{0>\delta_{\ell+k} >b_k^{-1}Z_{\ell+k}} + \I{b_k^{-1}Z_{\ell+k} > \delta_{\ell+k} > 0}}\I{\abs{\delta_{\ell+k}} \ge c\cdot a_k }} \\
		&= \E*{\p*{\I{0>\delta_{\ell+k} >b_k^{-1}Z_{\ell+k}} + \I{b_k^{-1}Z_{\ell+k} > \delta_{\ell+k} > 0}}\I{b_k^{-1}\abs{Z_{\ell+k}}\ge \abs{\delta_{\ell+k}} \ge c\cdot a_k}} \\
		&\le a_k^{1-q}b_k^{1-q}\E*{\abs{Z_{\ell+k}}^{q-1}\p*{\I{0>\delta_{\ell+k} >b_k^{-1}Z_{\ell+k}} + \I{b_k^{-1}Z_{\ell+k} > \delta_{\ell+k} > 0}}}\\
		&= a_k^{1-q}b_k^{1-q}\bigg (\int_0^\infty \int_0^{b_k^{-1}z}\abs{z}^{q-1}\rho_{\ell+k}(y, z)\text{d}y\text{d}z  + \int_{-\infty}^0 \int_{-b_k^{-1}z}^0\abs{z}^{q-1}\rho_{\ell+k}(y, z)\text{d}y\text{d}z  \bigg ).
	\end{aligned}
	\]
	By using the Taylor expansion \eqref{eqn:taylor-rho} and \Cref{assumpt:weak-err-st} we can obtain (assuming $\int_{\R}\rho_{\ell+k}(0,z)\text{d}z > 0$ as in the proof of \Cref{prop:str-err-strict-assumptions})
	\[
	\begin{aligned}
		\E*{\p*{\H{g} - \H{g_{\ell+k}}}^2\I{\eta_\ell = k}} &\lesssim a_k^{1-q}b_k^{-q}\\
		&= 2^{\frac{\gamma}{r}\p*{\ell + \theta\ell(1-r)}\p*{1-q} }2^{\p*{(q-1)\gamma/r - q\beta/2}(\ell+k)}.
	\end{aligned}
	\]
	When $r < \frac{2\gamma}{\beta}(1-1/q)$, this above term is dominant when $k=\ceil{\theta\ell}$. It follows that one can make the orders of $\Sigma_0$ and $\Sigma_1$ equal as $\ell\to\infty$ in \eqref{eqn:var-split} by taking $\theta$ as in \eqref{eqn:theta-strong}. When $\beta \ge \gamma$, instead we fix $\theta = 1$. A similar calculation to \Cref{thm:var-bound} then shows the result holds provided $r$ satisfies \eqref{eqn:rvals-strict}.
\end{proof}
\Cref{thm:var-bound} gives a larger value of $\theta$, allowing greater benefits from the refinement, when $\beta\le (q+1)\gamma/q$ and $q$ is bounded.
 
\subsection{Analysis of the Bias}
In the context of \Cref{thm:var-bound-strong}, \Cref{thm:complexity} implies the complexity of (adaptive) MLMC is affected by the convergence rate of $E_\ell$ whenever $\beta<\gamma$. To improve the rate given by \Cref{thm:standard-bias} due to adaptive sampling, we make a further assumption.
\begin{assumption}
	\label{assumpt:weak-err-adapt} 
	Define $Z_\ell, \beta>0$ as in \Cref{assumpt:q-moments} and $\frac{\beta}{2}\le\alpha\le \beta$ as in \Cref{assumpt:weak-err-nondiscont}. Then, for $j=0,1$ and all $\ell\in \mathbb{N},x\ge0$ we have 
	\[
		\abs*{\E*{\textnormal{sign}(Z_\ell)\abs{Z_\ell}^j\  \bigg{|}\ \abs{Z_\ell}\ge x}} \lesssim 2^{\ell\p*{\beta/2 - \alpha}}.
	\]
\end{assumption}
By \Cref{assumpt:weak-err-st,assumpt:weak-err-nondiscont}  we know that this condition holds for $j=1$ and $x=0$. \Cref{assumpt:weak-err-adapt} ensures that the mean of $Z_\ell$ converges at the same rate even when conditioned on taking large values. When $j=0$ the assumption implies that the probability of observing large positive $Z_\ell$  is reasonably close to the probability of observing large negative $Z_\ell$. The necessity for this assumption arises since the refined samples are only accepted before the maximum level if $\abs{\delta_\ell}$ is sufficiently large. As such, the error $\H{g} - \H{g_\ell}$ is non-zero only for suitably large values of $Z_\ell$. The resulting improvement to $E_\ell$ is discussed below.

\begin{lemma}
	\label{prop:weak-err-ad} 
	Let \Cref{assumpt:weak-err-st,assumpt:weak-err-adapt} hold for $\beta>0$ and $\frac{\beta}{2}\le\alpha\le\beta$. For $\beta\le\gamma$, if we tighten the bound on $r$ in \Cref{thm:var-bound-strong} to $r < \frac{2\gamma}{\beta}\frac{q-2}{q}$, then for $g_{\ell+\eta_{\ell}}$ as in \Cref{alg:ad-samp} and $\theta$ as in \eqref{eqn:theta-strong} we have 
	\[
		\abs*{\E*{\H{g} - \H{g_{\ell+\eta_\ell}}}} \lesssim 2^{-\alpha\p*{1+\theta}\ell}. 
	\]
\end{lemma}
\begin{proof}
	We bound 
	\[
		\begin{aligned}
			\abs*{\E*{\H{g} - \H{g_{\ell+\eta_\ell}}}} &\le \sum_{k=0}^{\ceil{\theta\ell}-1} \abs*{\E*{\p*{\H{g} - \H{g_{\ell+k}}}\I{\eta_\ell = k}}}+ \abs*{\E*{\H{g} - \H{g_{\ell+\ceil{\theta\ell}}}}}.
		\end{aligned}
	\]
	By \Cref{thm:standard-bias} we know that the final term satisfies 
	\begin{equation}
	\label{eqn:adbias-finalterm}
		\abs*{\E*{\H{g} - \H{g_{\ell+\ceil{\theta_\ell}}}}} \lesssim 2^{-\alpha(1+\theta)\ell}.
	\end{equation}
	By expanding the difference $\H{g} - \H{g_{\ell + k}}$ according to when the difference is either $\pm 1$ and bounding the event $\eta_\ell = k$ we arrive at
	\[
		\begin{aligned}
			&\abs*{\E*{\p*{\H{g} - \H{g_{\ell+k}}}\I{\eta_\ell = k}}}\\
			 &= \abs*{\E*{\p*{\I{b_{k}^{-1}Z_{\ell+k}<\delta_{\ell+k}<0} - \I{0<\delta_{\ell + k} < b_k^{-1}Z_{\ell+k}}}\I{\abs{\delta_{\ell+k}}\ge c\cdot a_k}}} \\
			&= \abs*{\E*{\p*{\I{b_{k}^{-1}Z_{\ell+k}<\delta_{\ell+k}<a_k} - \I{a_k<\delta_{\ell + k} < b_k^{-1}Z_{\ell+k}}}}} \\
			&= \abs*{\int_{-\infty}^{-b_ka_k}\int_{b_k^{-1}z}^{-a_k}\rho_{\ell+k}(y, z)\text{d}y\text{d}z - \int_{b_ka_k}^\infty \int_{a_k}^{b_k^{-1}z}\rho_{\ell+k}(y, z)\text{d}y\text{d}z},
		\end{aligned}
	\]
	where $a_k, b_k$ are as in the proof of \Cref{thm:var-bound}. We again use the Taylor expansion \eqref{eqn:taylor-rho} on the density $\rho_{\ell+k}(y, z)$. The absolute value of the zero'th-order term is (assuming $\int_{\R}\rho_{\ell+k}(0,z)\text{d}z > 0$ as in the proof of \Cref{prop:str-err-strict-assumptions})
	\[
		\begin{aligned}
			&\abs*{\int_{-\infty}^{-b_ka_k} \p*{-a_k - b_k^{-1}z}\rho_{\ell+k}(0, z)\text{d}z + \int_{b_ka_k}^\infty \p*{a_k-b_k^{-1}z}\rho_{\ell+k}(0, z)\text{d}z} \\
			&\le a_k\abs*{\E*{\text{sign}(Z_{\ell+k})\I{\abs{Z_{\ell+k}}\ge a_kb_k}|\delta_{\ell + k}=0}} + b_k^{-1}\abs{\E*{Z_{\ell+k}\I{\abs{Z_{\ell+k}}\ge b_ka_k}| \delta_{\ell+k}=0}}\\
			&\lesssim \prob*{\abs*{Z_{\ell+k}}\ge b_ka_k}\bigg (a_k\abs*{\E*{\text{sign}(Z_{\ell+k})\ \bigg{|} \ \abs*{Z_{\ell+k}}\ge b_ka_k, \delta_{\ell + k}=0} }\\
			&\qquad\qquad\qquad + b_k^{-1}\abs*{\E*{Z_{\ell+k}\ \bigg{|}\ \abs{Z_{\ell+k}}\ge b_ka_k, \delta_{\ell + k}=0}}\bigg )\\
			&\lesssim a_k^{1-q}b_k^{-q}2^{\p*{\ell+k}\p*{\beta/2 - \alpha}}
		\end{aligned}
	\]
	where we used \Cref{assumpt:weak-err-adapt} and bounded $\prob*{\abs*{Z_{\ell+k}}\ge b_ka_k} = \E*{\I{\abs{Z_k}\ge b_ka_k}}\le a_k^{-q}b_k^{-q}\E*{\abs{Z_k}^q}$. For the first-order term, we obtain 
	\[
	\begin{aligned}
		&\abs*{a_k^2\int_{-\infty}^{\infty}\I{\abs{z}\ge b_ka_k}\frac{\partial}{\partial y}\rho_{\ell+k}(0, z)\text{d}z - b_k^{-2}\int_{-\infty}^{\infty}\I{\abs{z}\ge b_ka_k}z^2\frac{\partial}{\partial y}\rho_{\ell+k}(0, z)\text{d}z}\\
		&\le \abs*{b_k^{-q}a_k^{2-q}\int_{-\infty}^{\infty}\abs{z}^q\frac{\partial}{\partial y}\rho_{\ell+k}(0, z)\text{d}z - b_k^{-2-q}a_k^{-q}\int_{-\infty}^{\infty}\abs{z}^{2+q}\frac{\partial}{\partial y}\rho_{\ell+k}(0, z)\text{d}z}\\
		&\lesssim b_k^{-q}a_k^{2-q}
	\end{aligned}
	\]
	by \Cref{assumpt:weak-err-st}.  Similarly, we can bound the second-order term up to a constant by $a_k^{3-q}b_k^{-q}$. Consequently, we have 
	\[
		 \sum_{k=0}^{\ceil{\theta\ell}-1} \abs*{\E*{\p*{\H{g} - \H{g_{\ell+k}}}\I{\eta_\ell = k}}} \lesssim \sum_{k=0}^{\ceil{\theta\ell}-1} a_k^{1-q}b_k^{-q}2^{\p*{\ell+k}\p*{\beta/2 - \alpha}} + \sum_{k=0}^{\ceil{\theta\ell}-1}a_k^{2-q}b_k^{-q}.
	\]
	Provided $r < \frac{2\gamma}{\beta}\frac{q-2}{q}$, the dominant cost of each sum on the right hand side occurs at $k = \ceil{\theta\ell} - 1$, giving 
	\[
	\begin{aligned}
		\sum_{k=0}^{\ceil{\theta\ell}-1} \abs*{\E*{\p*{\H{g} - \H{g_{\ell+k}}}\I{\eta_\ell = k}}} &\lesssim a_{\theta\ell}^{1-q}b_{\theta\ell}^{-q}2^{\ell\p*{1+\theta}\p*{\beta/2 - \alpha}} + a_{\theta\ell}^{2-q}b_{\theta\ell}^{-q}\\
		&\lesssim 2^{-\alpha\p{1+\theta}\ell} + 2^{-\beta\p{1+\theta}\ell},
	\end{aligned}
	\]
	for $\theta$ as in \eqref{eqn:theta-strong}.
\end{proof}
Numerical tests suggest that the previous result does not hold when \Cref{assumpt:weak-err-adapt} is false, see \Cref{app:weakerr}. However, one can still obtain reasonable convergence rates of $E_\ell$ without this result by \Cref{remark:weakbound}. 
\subsection{Bounds on $\work$}
We conclude this section with a discussion on how the improved variance rate given by \Cref{thm:var-bound} affects the work bounds of MLMC. We begin by discussing the impact of adaptive sampling under the weaker assumptions. 

\begin{theorem}
	\label{corr:ad-comp} 
	Under the assumptions of \Cref{thm:var-bound}, the total work of MLMC using adaptive sampling as in \Cref{alg:ad-samp} with $\dhl$ given by \eqref{eqn:dh-ad}  is 
	\[\work\lesssim
	\begin{cases}
		\tol^{-2} & \beta > \p{\frac{q+1}{q}}\gamma \\
		\tol^{-2}\p*{\log\tol}^2& \beta = \p{\frac{q+1}{q}}\gamma \\
		\tol^{-2\p{\frac{q+1}{q}}\gamma/\beta} & \beta < \p{\frac{q+1}{q}}\gamma 
	\end{cases}.
	\] 
\end{theorem}
\begin{proof}
	The result follows from applying  \Cref{prop:work-rate}, \Cref{thm:var-bound} and \Cref{remark:weakbound} to \Cref{thm:complexity}, with $\dhl$ given by \eqref{eqn:dh-ad}.
\end{proof}
This result should be contrasted with \Cref{cor:st-comp-weak}. In particular, note how the canonical $\tol^{-2}$ complexity is obtained when $\beta > \p{\frac{q+1}{q}}\gamma$ as opposed to when $\beta > 2\p{\frac{q+1}{q}}\gamma$ for non-adaptive sampling. Moreover, even in the sub-optimal case when $\beta<\p{\frac{q+1}{q}}\gamma$ the bound is improved by a factor of $\tol^{-1}$ over the non-adaptive case. In some cases, \Cref{assumpt:q-moments} holds for all $q<\infty$ and the above result holds in the limit $q\to\infty$ provided one adds a factor $-\nu$ for any $\nu>0$ to the rate whenever $\beta\le \gamma$. When the assumptions of \Cref{prop:weak-err-ad} hold, we obtain a slightly stronger result.
\begin{theorem}
	\label{cor:ad-comp-strong}
	Under \Cref{assumpt:weak-err-st} and, if $\beta < \gamma$, under \Cref{assumpt:weak-err-adapt} the total work of MLMC using adaptive sampling as in \Cref{alg:ad-samp} with $\dhl$ given by \eqref{eqn:dh-ad} is 
	\[
	\work \lesssim 
	\begin{cases}
		\tol^{-2} & \beta > \gamma \\
		\tol^{-2}\p*{\log\tol}^2 & \beta = \gamma \\
		\tol^{-2-\p*{1-\beta/\p{2\gamma}}\p*{\gamma-\beta}/\alpha} & \beta < \gamma 
	\end{cases}.
	\]
\end{theorem}
\begin{proof}
	The result follows immediately by combining \Cref{prop:work-rate} and \Cref{thm:var-bound} with \Cref{thm:complexity} for $\dhl$ given by \eqref{eqn:dh-ad}. When $\beta < \gamma$ we use \Cref{prop:weak-err-ad} to obtain a rate for $E_\ell$.
\end{proof}
The previous result should be compared with \Cref{cor:st-comp-str}. Again, we can see optimal complexities for $\beta$ half as large as in the non-adaptive case. When $\beta<\gamma$ and $\alpha = \beta$ we can observe an improvement of order $\tol^{-0.5}$ in the complexity.

\section{Numerical Experiments}
\label{sect:num-ex}
This section presents several numerical experiments to highlight the preceding theory\footnote{The code used for these experiments is written in Python, and can be found at \url{https://github.com/JSpence97/mlmc-for-probabilities}.}. We begin with some remarks on the technical components of MLMC. 

\subsubsection*{Optimal Starting Level}
In \Cref{sect:mlmc-probs} we consider the MLMC estimator starting at level $\ell =0$. When the approximations $g_\ell$ have pre-asymptotic behavior at small levels, it may be more efficient to start from some level $\ell_0 > 0$. For adaptive sampling, this is not the same as simply adjusting the work required at level 0 by a constant to account for a more accurate starting estimator. To see this, observe from \Cref{alg:ad-samp} that samples at level $\ell =0$ cannot be refined further. In contrast, at level $\ell_0 > 0$ samples can be refined to maximum level $\ell_0 +\ceil{\theta\ell_0}$. A heuristic approach for estimating the optimal starting level by a small computation is given in \cite[Section 3]{GilesHajiAli:2018}. We use optimal starting levels to obtain all MLMC estimates in the following sections.

\subsubsection*{Error Estimation}
We illustrate the results of previous sections using  the average work of sampling the multilevel correction term, $W_\ell$, and the multilevel correction variance $V_\ell$ \eqref{str-err} and bias $E_\ell$ \eqref{weak-err}. Typically,  $V_\ell$ and  $E_\ell$ must be estimated using Monte Carlo sampling within MLMC. The robustness and accuracy of standard MLMC algorithms \cite{Giles2008MLMC, giles:2015review} depends on reliable estimates of $V_\ell, E_\ell$ to determine the optimal final level $L^*$ and number of samples per level $\{M_\ell^* \}_{\ell_0\le\ell\le L^*}$  required to have mean square error $\tol^2$. For example, let
\[
	\dh =
	\begin{cases}
		1 & \text{with probability } p_\ell\\
		-1 & \text{with probability } q_\ell\\
		0 &  \text{with probability } 1 - p_\ell - q_\ell 
	\end{cases},
\]
so that $E_\ell = \abs{p_\ell-q_\ell}$ and $V_\ell = p_\ell+q_\ell - (p_\ell - q_\ell)^2$. Hence, estimation of $E_\ell, V_\ell$ requires good estimates of $p_\ell, q_\ell$. However, by \Cref{assumpt:q-moments} we have $p_\ell, q_\ell\to 0$ as $\ell\to \infty$. Thus, we require more samples to reliably estimate $p_\ell, q_\ell$ as $\ell$ increases, which contradicts the intuition that MLMC aims to reduce the number of samples required at the finest levels. As a result, the robustness of MLMC can be affected by poor parameter estimation at the finest levels. One approach to estimate $p_\ell, q_\ell$ is detailed in \cite{Elfverson:2016selectiverefinement}, using Bayesian estimation with a beta prior distribution. An alternative solution, and the one used for the results stated here, is to estimate the proportionality constants in the bounds on $V_\ell, E_\ell$. We estimate these constants using the continuation MLMC approach discussed in \cite{Collier:CMLMC}.  \subsection{Nested Expectation}\label{sect:nest-ex}
The first numerical experiment is concerned with multilevel nested simulation as in \cite{GilesHajiAli:2018, GilesHajiAli:2019sampling}. We take $g = \E*{X|Y}$ so that \eqref{eqn:heaviside-problem} becomes 
\(
	\E*{\H{\E*{X|Y}}}.
\) 
Approximations of $g$ at a level $\ell$ are given by an inner Monte Carlo estimator
\begin{equation}\label{eqn:gell-nest}
	g_\ell = \frac{1}{N_\ell}\sum_{n=1}^{N_\ell} X^{(n)}(Y), \qquad X^{(n)}(Y)\iid X|Y
\end{equation}
using $N_\ell = N_02^{\gamma\ell}$ samples. When refining from level $k$ to $k+1$ in \Cref{alg:ad-samp} we take the $N_k$ samples used to sample $g_k$ and add another $N_k\p*{2^\gamma-1}$ independent samples to form the sample of $g_{k+1}$. We assume $\sigma_\ell^2$ is given by the sample variance of the samples used to generate $g_\ell$,  other choices of $\sigma_\ell$ are discussed in \Cref{app:constsigma}. We do not rigorously prove that \Cref{assumpt:q-moments,assumpt:delta-bound} hold, but provide motivation for their validity with $\beta=\gamma$. The discussion is supported by numerical experiments at the end of this section.

Defining
\[
\begin{aligned}
	S_{N_\ell}(Y) &\defeq \sum_{n=1}^{N_\ell} \p*{X^{(n)}(Y) - \E*{X|Y}}\\
	V_{N_\ell}^2(Y) &\defeq \sum_{n=1}^{N_\ell}\p*{X^{(n)}(Y) - \E*{X|Y}}^2,
\end{aligned}
\]
we can express $Z_\ell$ \eqref{eqn:Z_ell} as
\begin{equation}
\label{eqn:Z-nest}
\abs{Z_\ell} = N_0^{-1/2}\abs*{\frac{S_{N_\ell}(Y)}{V_{N_\ell}(Y)}}\frac{1}{\sqrt{1-N_\ell^{-1}\abs*{S_{N_\ell}(Y)/V_{N_\ell}(Y)}^2}}.
\end{equation}
The ratio $\frac{S_{N_\ell}(Y)}{V_{N_\ell}(Y)}$ is known as a self normalised sum \cite{pena:selfnormalizedsums}. In particular, by the strong law of large numbers, $\abs{Z_\ell}/\abs*{\frac{S_{N_\ell}(Y)}{V_{N_\ell}(Y)}}\to N_0^{-1/2}$ almost surely as $\ell\to\infty$ \cite[Proposition 1]{Griffin:t-tightness}. Provided the self normalised sum is stochastically bounded uniformly in $Y$ \cite[Corollary 2.10]{Gine:asymptotic-t}, we have \cite[Theorem 2.5]{Gine:asymptotic-t} 
\[
\sup_\ell \E*{\abs*{\frac{S_{N_\ell}(Y)}{V_{N_\ell}(Y)}}^q} < \infty, \quad \text{ for all } q<\infty.
\]
We can then use the asymptotic equivalence of $\abs*{\frac{S_{N_\ell}(Y)}{V_{N_\ell}(Y)}}$ and $\abs{Z_\ell}$ and H\"olders inequality to show that for any $q<\infty$ there exists $N_0\ge 0$ such that  \Cref{assumpt:q-moments} holds with $g_\ell$ given by \eqref{eqn:gell-nest} for $\beta = \gamma$. Note that $N_0$ may depend on $q$. For example, Student's $t$-distribution can be written in the form \eqref{eqn:Z-nest} up to a constant factor, and has finite $q$-moments only for $q<N_\ell$. \Cref{assumpt:delta-bound} would follow by assuming a similar condition on the distribution of $\abs{g}/\sqrt{\var{X|Y}}$ as in \cite{GilesHajiAli:2018}. We leave a rigorous justification of this fact to future work.\\

For comparison with \cite{GilesHajiAli:2018} we consider the model problem used there, given by
\[
X = \frac{2}{100}\p*{Y^2 - Y_0^2} + \frac{7\sqrt{2}}{25}YY_1 - 0.0805
\]
for $Y, Y_0, Y_1 \iid \mathcal{N}(0,1)$. For this problem, one has $\E*{\H{\E*{X|Y}}} \approx 0.025$. In \cite{GilesHajiAli:2018} the use of additional measures such as antithetic sampling of $\dhl$ is considered to reduce the total cost of MLMC by a constant factor independent of the error bound $\tol$. Such approaches can easily be altered to suit the present setup. We emphasize that the key difference between \Cref{alg:ad-samp} and the adaptive scheme in \cite{GilesHajiAli:2018} for this setup is that here we do not re-sample all values of $X^{(n)}(Y)$ when refining to higher levels and the samples generated in \Cref{alg:ad-samp} are used to form our estimate of $g_{\ell+\eta_\ell}$, in contrast to \cite{GilesHajiAli:2018,GilesHajiAli:2019sampling}.

The MLMC estimator is computed using non-adaptive sampling with $\gamma = 1, 2$ and adaptive sampling as in \Cref{alg:ad-samp} with $\gamma =1$ and  $r = 1.95, \theta = 1$ to fulfill the assumptions of \Cref{prop:work-rate} and \Cref{thm:var-bound} in the limit $q\to \infty$. The confidence constant is taken to be $c=3/\sqrt{N_0}$, which aligns with the corresponding parameter in \cite{GilesHajiAli:2018}. For each method, we plot $W_\ell, V_\ell$ and $E_\ell$ versus $\ell$.  Results are shown in \Cref{fig:nestexpect}. The top left plot shows $W_\ell$ vs $\ell$. By construction, the work per level for the non-adaptive schemes is a deterministic term proportional to $2^{\gamma\ell}$. For the adaptive scheme and $\ell>2$, we observe $W_\ell\propto 2^{\ell}$, increased by a constant factor over the non-adaptive sampler with $\gamma=1$. This agrees with \Cref{prop:work-rate}, which states that adaptive sampling does not affect the rate at which $W_\ell$ increases. The variance $V_\ell$ per level is shown in the top right plot of \Cref{fig:nestexpect}. Following from \Cref{thm:standard-sampling-var} with $q\to\infty$, the non-adaptive samplers have variance decreasing at rate  $\beta/2 \approx \gamma/2$. Instead, the adaptive sampler matches the variance seen for the non-adaptive method with $\gamma = 2$, as predicted by \Cref{thm:var-bound}. Moreover, in the bottom left plot of \Cref{fig:nestexpect}, we see that the bias reduction rates guaranteed from \Cref{thm:standard-bias} and \Cref{prop:weak-err-ad} with $\alpha = \beta$.  In other words, the adaptive scheme exhibits the same variance and bias reduction rate as the non-adaptive method with $\gamma = 2$, but has expected work per level comparable to the non-adaptive method with $\gamma = 1$.

In the bottom right plot of \Cref{fig:nestexpect}, we display the total work of sampling $\mlmc$ multiplied by $\tol^2$ against the accuracy $\tol$, normalised according to the true value 0.025. The total work is taken as the number of inner samples generated from $X$ for a given $Y$. For each method, we run the algorithm from an estimated optimal starting level as in \cite{GilesHajiAli:2018}. The theoretical complexity rates given by \Cref{cor:st-comp-str} and \Cref{corr:ad-comp,cor:ad-comp-strong} for $q\to\infty$ are plotted as dashed and dotted lines, highlighting the applicability of the preceding theory to this example. At a normalised error of around $10^{-2.5}$, we observe a reduction in cost by a factor of around 7 for adaptive sampling, which is roughly the same as seen in \cite{GilesHajiAli:2018}.  

\begin{figure}[h]
	\centering
	\begin{tabular}{ccc}
		\ref{plot:nestmpgam1} $\gamma = 1$ & \ref{plot:nestmpgam2} $\gamma = 2$ & \ref{plot:nestmpad} Adaptive
	\end{tabular} \\
\begin{tabular}{cc}
	\ref*{plot:nest-log} $\tol^{-2}\p*{\log\tol}^2$ & \ref*{plot:nest-eps}  $\tol^{-2.5}$
\end{tabular}\\
	\begin{tikzpicture}
		\begin{axis}[
		xlabel = $\ell$,
		ylabel = $W_\ell$,
		grid = both,
		ymode = log,
		]
		\addplot+[mark = asterisk, mark size = 3pt, black,] 
		table[x = level, y expr = \thisrow{cost}/\thisrow{M}, col sep = space,]{results/data-nest-mp/ad-levels.csv};
		\label{plot:nestmpad}
		
		\addplot+[mark = diamond, mark size = 3pt, black] 
		table[x = level, y expr = 32*\thisrow{cost}/\thisrow{M}, col sep = space,]{results/data-nest-mp/st-gam1-levels.csv};
		\label{plot:nestmpgam1}
		
		\addplot+[mark = o, mark size = 3pt, black] 
		table[x = level, y expr = 32*\thisrow{cost}/\thisrow{M}, col sep = space,]{results/data-nest-mp/st-gam2-levels.csv};
		\label{plot:nestmpgam2}	
		\end{axis}
		
	\end{tikzpicture}
	\begin{tikzpicture}
	\begin{axis}[
	xlabel = $\ell$,
	ylabel = $V_\ell$,
	grid = both,
	ymode = log,
	]
	\addplot[mark = asterisk, mark size = 3pt] 
	table[x = level, y expr = \thisrow{V}, col sep = space,]{results/data-nest-mp/ad-levels.csv};

	\addplot[mark = diamond, mark size = 3pt] 
	table[x = level, y expr = \thisrow{V}, col sep = space,]{results/data-nest-mp/st-gam1-levels.csv};
	
	\addplot[mark = o, mark size = 4pt] 
	table[x = level, y expr = \thisrow{V}, col sep = space,]{results/data-nest-mp/st-gam2-levels.csv};
	
	\draw(3.2, 0.0002) node[fill = white,]{$\Order{2^{-\ell}}$};
	\addplot[domain = 2:6]{0.01*2^(-x)};
	
	\draw(4.7, 0.04) node[fill = white]{$\Order{2^{-\ell/2}}$};
	\addplot[domain = 2:6]{0.075*2^(-x/2)};

	\end{axis}
	\end{tikzpicture}
	\\
	\begin{tikzpicture}
	\begin{axis}[
	xlabel = $\ell$,
	ylabel = $E_\ell$,
	grid = both,
	ymode = log,
	]
	\addplot[mark = asterisk, mark size = 3pt] 
	table[x = level, y expr = abs(\thisrow{m}), col sep = space,]{results/data-nest-mp/ad-levels.csv};

	\addplot[mark = diamond, mark size = 3pt] 
	table[x = level, y expr = abs(\thisrow{m}), col sep = space,]{results/data-nest-mp/st-gam1-levels.csv};

	\addplot[mark = o, mark size = 3pt] 
	table[x = level, y expr = abs(\thisrow{m}), col sep = space,]{results/data-nest-mp/st-gam2-levels.csv};
	
	\draw(3.2, 0.000004) node[fill = white,]{$\Order{2^{-2\ell}}$};
	\addplot[domain = 2:6]{0.008*2^(-2*x)};
	
	\draw(4.7, 0.04) node[fill = white]{$\Order{2^{-\ell}}$};
	\addplot[domain = 2:6]{0.2*2^(-x)};
	\end{axis}
	\end{tikzpicture}
	\begin{tikzpicture}
	\begin{axis}[
	xlabel = Normalised $\tol$,
	ylabel = $\textnormal{Total Work}\times\tol^2$,
	grid = both,
	ymode = log,
	xmode = log,
	legend pos = north east,
	legend columns = 2,
	ymax = 20000,
	]
	\addplot[style = dashed, domain = 0.003:0.1,]{35*ln(x)*ln(x)};
	\label{plot:nest-log}
	
	\addplot[style = densely dotted, domain = 0.003:0.1,]{500*x^(-0.5)};
	\label{plot:nest-eps}

	\addplot[mark = asterisk, mark size = 3pt] 
	table[x expr = \thisrow{tol}/0.025, y expr = \thisrow{cost}*\thisrow{tol}*\thisrow{tol}, col sep = space,]{results/data-nest-mp/ad-mlmc.csv};

	\addplot[mark = diamond, mark size = 3pt] 
	table[x expr = \thisrow{tol}/0.025, y expr = 32*\thisrow{cost}*\thisrow{tol}*\thisrow{tol}, col sep = space,]{results/data-nest-mp/st-gam1-mlmc.csv};

	\addplot[mark = o, mark size = 3pt] 
	table[x expr = \thisrow{tol}/0.025, y expr = 32*\thisrow{cost}*\thisrow{tol}*\thisrow{tol}, col sep = space,]{results/data-nest-mp/st-gam2-mlmc.csv};
	\end{axis}
	\end{tikzpicture}
	\caption{Results for the nested simulation model problem: Expected work per level $W_\ell$ (top left) taken as expected number of required inner samples from $X|Y$, multilevel correction variance $V_\ell$ (top right) and bias $E_\ell$ (bottom left) versus $\ell$. The total work of MLMC times $\tol^2$ versus $\tol$, normalised by the true value of the solution (bottom right). Results are  given for non-adaptive schemes with $\gamma =1, 2$ and the adaptive scheme with $r=1.95$. }
	\label{fig:nestexpect}
\end{figure}
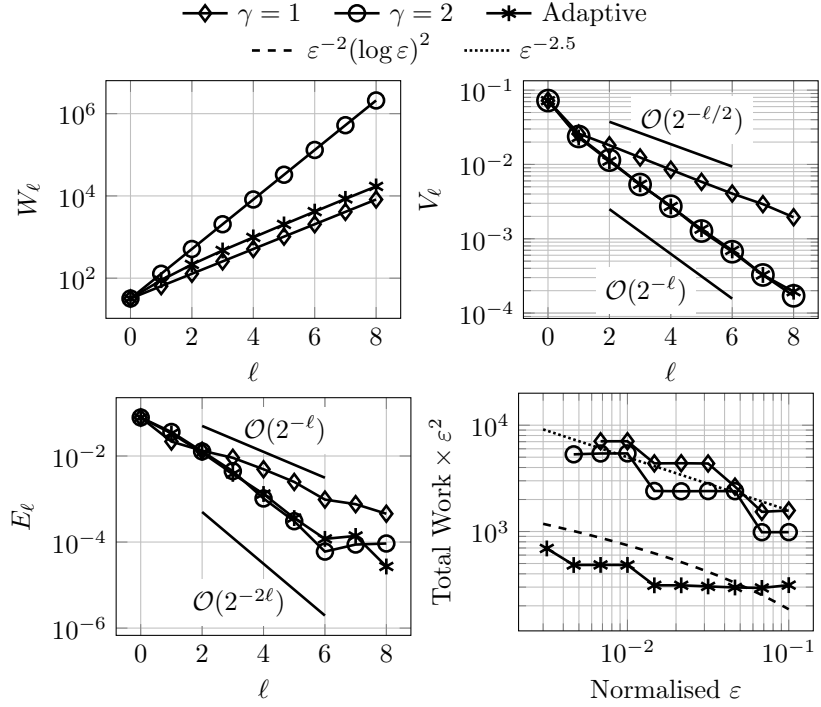  \subsection{Stochastic Differential Equations}
We now consider a setup where $g$ is determined by $d$ stock prices modeled by the geometric Brownian motions
\begin{equation}\label{eqn:sdes}
\text{d}S^{\p{i}}\p{t} = a_iS^{(i)}(t)\text{d}t + b_{i}S^{(i)}(t)\text{d}W^{\p{i}}\p{t}, \quad 1\le i\le d,
\end{equation}
where the one-dimensional Wiener processes take the form
\[
W^{(i)}(t) = \rho W_{\text{com}}(t) + \sqrt{1-\rho^2}W_{\text{ind}}^{(i)}(t),
\]
for a correlation coefficient $\rho\in [0,1]$ and independent  Wiener processes $W_{\text{com}}(t)$ and $  \{W_{\text{ind}}^{(i)}(t) \}_{i=1}^{d}$. Here, $W_\text{com}(t)$ models common market noise shared by all of the stocks whereas $W_\text{ind}^{(i)}(t)$ represents idiosyncratic noise of stock $i$ only. Specifically, we set
\begin{equation}\label{eqn:sde-prob}
g = \frac{1}{d}\sum_{i=1}^{d} S^{(i)}(1) - K,
\end{equation}
so that $\prob{g>0}$ reflects the non-discounted price of a so-called digital option, a financial derivative which pays a unit price at time 1 if the mean value of the stocks exceeds $K$, and nothing otherwise. We assume $0.05\le a_i\le0.15, 0.01\le b_i\le0.4$ and $0.9\le S^{(i)}(0)\le 1.1$ are constant. Unless otherwise stated, we uniformly sample each of these parameters before the MLMC computation. \\

The approximate samples, $g_\ell$, are computed using either Euler-Maruyama or Milstein discretisation of the underlying SDEs with step size $h_\ell \lesssim 2^{-\gamma\ell}$. When adaptively refining samples of $g_\ell$ we use the Brownian Bridge construction to refine the sampled Wiener paths conditioned on their existing points \cite[Section 1.8]{Kloeden:1999}. Specifically, given $W_{nh_\ell}$ and $W_{(n+1)h_\ell}$ we can sample the Wiener process at time $\p*{n+1/2}h_\ell$ using
\[
	W_{\p*{n+1/2}h_\ell} \overset{\text{d}}{=} \frac{W_{nh_\ell} + W_{(n+1)h_\ell}}{2} + \sqrt{\frac{h_\ell}{4}}\zeta_n, \quad \text{for } \zeta_n \iid \mathcal{N}(0,1).
\]
This procedure can be used recursively to refine from step-size $h_k$ to $h_{k+1} = 2^{-\gamma}h_k$ within adaptive MLMC. It follows from the strong convergence results of each method that \Cref{assumpt:q-moments} holds for all $q<\infty$ using deterministic, constant $\sigma_\ell$ and $\beta = \gamma$ for Euler-Maruyama \cite[Theorem 10.2.2]{Kloeden:1999} and $\beta = 2\gamma$ for the Milstein scheme \cite[Theorem 10.3.5]{Kloeden:1999}. That \Cref{assumpt:delta-bound} holds for constant $\sigma_\ell =\sigma$ can be shown for Euler-Maruyama using \cite[Theorem 2.3]{Gobet:2008} to bound the difference in the densities of $g_\ell$ and $g$. The result then follows since $g$ has a bounded density \cite[Theorem 10.9.11]{kuo:stoch-integration}. Moreover, from the weak convergence results in \cite{Kloeden:1999}, we know that \Cref{assumpt:weak-err-nondiscont} holds for both SDE schemes with $\alpha = \beta$. Bounding the variance of $S^{(i)}(1)$ for all instances of $a_i, b_i, S^{(i)}(0)$ we see that $\var{g} \propto d^{-1}$. Consequently, we choose $\sigma_\ell = \sigma = d^{-1/2}$.  \\

We first consider \eqref{eqn:sde-prob} for a single stock, $d=1$. In \eqref{eqn:sdes}, we take $a_1 = 0.05, b_1 = 0.4$ and $K$ is chosen such that $\E*{\H{g}} = 0.025$. The terms $W_\ell, V_\ell, C_\ell$ of the MLMC estimator is shown in \Cref{fig:sde-1d-levels} for Euler and Milstein approximation of $g$, using non-adaptive and adaptive simulation. For non-adaptive sampling we consider the cases $h_\ell = 2^{-\gamma\ell}$ for $\gamma = 1,2$. The adaptive samplers take $\gamma = \theta = c = 1$. For the Euler-Maruyama scheme we take $r=1.95$. Since $\beta = 2\gamma$ for the Milstein scheme, \Cref{thm:var-bound} allows us to take larger values of $r$ and we set $r = 10$ here.  $W_\ell$ is taken as the expected number of SDE steps required from the fine and coarse estimator at level $\ell$. By construction, the work for both non-adaptive samplers is proportional to $2^{\gamma\ell}$. The adaptive schemes have $W_\ell \lesssim 2^{\ell}$ following \Cref{prop:work-rate}. Note that the expected work per sample is slightly lower at each level for adaptive sampling using the Milstein scheme, as the larger value $r = 10$ requires fewer refinements to be made. For the Euler-Maruyama samplers we see $V_\ell \lesssim 2^{-\gamma\ell/2}$ for the non-adaptive and $V_\ell\lesssim 2^{-\ell}$ for the adaptive sampler, as expected for $\beta=\gamma$. These bounds are all squared when using the Milstein scheme since $\beta = 2\gamma$ in this case. Moreover, we observe $E_\ell \lesssim 2^{-\gamma\ell}$ for the non-adaptive sampler for both SDE schemes, with $E_\ell \lesssim 2^{-2\ell}$ for the adaptive samplers. This provides evidence that the stronger results following from \Cref{assumpt:weak-err-st,assumpt:weak-err-nondiscont,assumpt:weak-err-adapt} hold for the Euler-Maruyama scheme. For the Milstein scheme, the observed rates of $E_\ell$ follow immediately from the equivalent rates on $V_\ell$ and \Cref{remark:weakbound}.
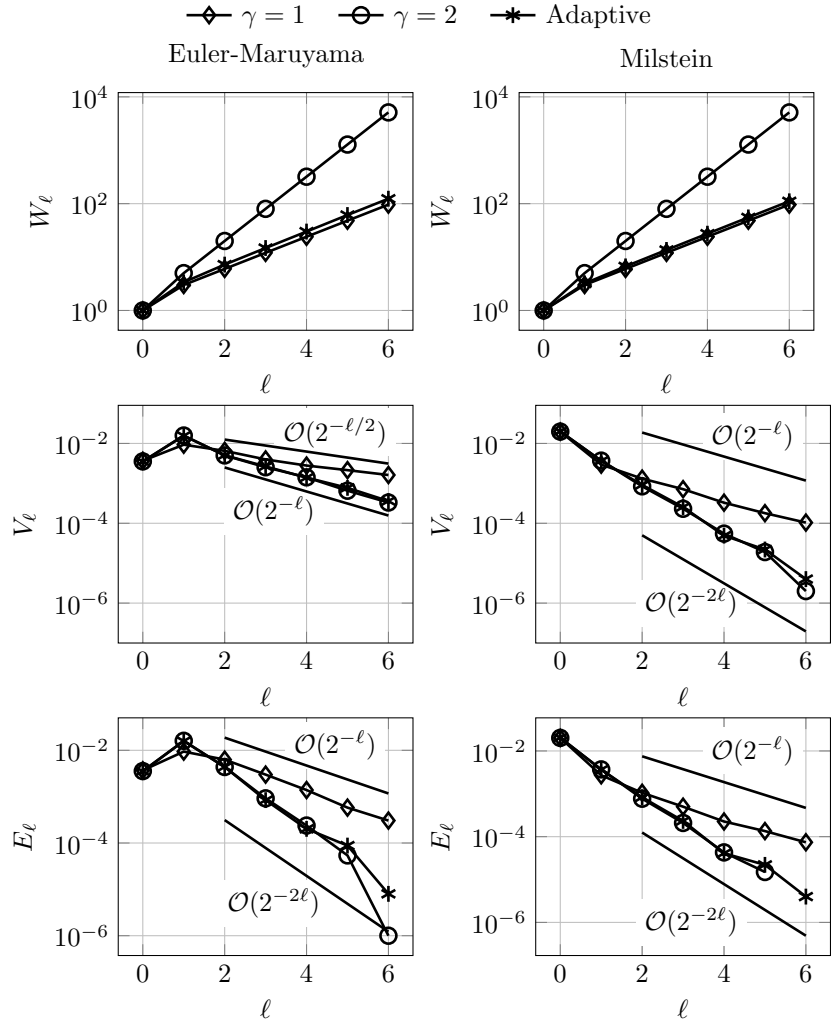
\begin{figure}
	\centering
	\begin{center}
	\begin{tabular}{ccc}
		\ref*{plot:gam1-euler} $\gamma = 1$ & \ref*{plot:gam2-euler} $\gamma = 2$ & \ref*{plot:ad-euler} Adaptive
	\end{tabular}\\
	\begin{tikzpicture}
		\begin{axis}[
		xlabel = $\ell$,
		ylabel = $W_\ell$,
		grid = both,
		ymode = log,
		title = Euler-Maruyama
		]
		\addplot+[mark = asterisk, mark size = 3pt, black,] 
		table[x = level, y expr = \thisrow{cost}/\thisrow{M}, col sep = space,]{results/data-sum-strike/ad-euler-levels.csv};
		\label{plot:ad-euler}
		
		\addplot+[mark = diamond, mark size = 3pt, black] 
		table[x = level, y expr = \thisrow{cost}/\thisrow{M}, col sep = space,]{results/data-sum-strike/st-gam1-euler-levels.csv};
		\label{plot:gam1-euler}
		
		\addplot+[mark = o, mark size = 3pt, black] 
		table[x = level, y expr = \thisrow{cost}/\thisrow{M}, col sep = space,]{results/data-sum-strike/st-gam2-euler-levels.csv};
		\label{plot:gam2-euler}	
		\end{axis}
	\end{tikzpicture}
	\begin{tikzpicture}
	\begin{axis}[
	xlabel = $\ell$,
	ylabel = $W_\ell$,
	grid = both,
	ymode = log,
	title = Milstein
	]
	\addplot+[mark = asterisk, mark size = 3pt, black,] 
	table[x = level, y expr = \thisrow{cost}/\thisrow{M}, col sep = space,]{results/data-sum-strike/ad-milstein-levels.csv};
	\label{plot:ad-mil}
	
	\addplot+[mark = diamond, mark size = 3pt, black] 
	table[x = level, y expr = \thisrow{cost}/\thisrow{M}, col sep = space,]{results/data-sum-strike/st-gam1-milstein-levels.csv};
	\label{plot:gam1-mil}
	
	\addplot+[mark = o, mark size = 3pt, black] 
	table[x = level, y expr = \thisrow{cost}/\thisrow{M}, col sep = space,]{results/data-sum-strike/st-gam2-milstein-levels.csv};
	\label{plot:gam2-mil}	
	\end{axis}
	\end{tikzpicture}\\
	\begin{tikzpicture}
	\begin{axis}[
	xlabel = $\ell$,
	ylabel = $V_\ell$,
	grid = both,
	ymode = log,
	ymax = 0.09,
	ymin = 0.0000001
	]
	\draw(3.2, 0.0002) node[fill = white,]{$\Order{2^{-\ell}}$};
	\addplot[domain = 2:6]{0.01*2^(-x)};
	
	\draw(4.7, 0.019) node[fill = white]{$\Order{2^{-\ell/2}}$};
	\addplot[domain = 2:6]{0.025*2^(-x/2)};
	\addplot+[mark = asterisk, mark size = 3pt, black,] 
	table[x = level, y expr = \thisrow{V}, col sep = space,]{results/data-sum-strike/ad-euler-levels.csv};

	\addplot+[mark = diamond, mark size = 3pt, black] 
	table[x = level, y expr = \thisrow{V}, col sep = space,]{results/data-sum-strike/st-gam1-euler-levels.csv};

	\addplot+[mark = o, mark size = 3pt, black] 
	table[x = level, y expr = \thisrow{V}, col sep = space,]{results/data-sum-strike/st-gam2-euler-levels.csv};

	\end{axis}
	\end{tikzpicture}
	\begin{tikzpicture}
	\begin{axis}[
	xlabel = $\ell$,
	ylabel = $V_\ell$,
	grid = both,
	ymode = log,
	ymax = 0.09,
	ymin = 0.0000001
	]
	\addplot+[mark = asterisk, mark size = 3pt, black,] 
	table[x = level, y expr = \thisrow{V}, col sep = space,]{results/data-sum-strike/ad-milstein-levels.csv};

	\addplot+[mark = diamond, mark size = 3pt, black] 
	table[x = level, y expr = \thisrow{V}, col sep = space,]{results/data-sum-strike/st-gam1-milstein-levels.csv};

	\addplot+[mark = o, mark size = 3pt, black] 
	table[x = level, y expr = \thisrow{V}, col sep = space,]{results/data-sum-strike/st-gam2-milstein-levels.csv};
	
	\draw(3.2, 0.000001) node[fill = white,]{$\Order{2^{-2\ell}}$};
	\addplot[domain = 2:6]{0.0008*2^(-2*x)};
	
	\draw(4.7, 0.015) node[fill = white]{$\Order{2^{-\ell}}$};
	\addplot[domain = 2:6]{0.075*2^(-x)};
	\end{axis}
	\end{tikzpicture}
	\begin{tikzpicture}
	\begin{axis}[
	xlabel = $\ell$,
	ylabel = $E_\ell$,
	grid = both,
	ymode = log,
	]
	\addplot[mark = asterisk, mark size = 3pt] 
	table[x = level, y expr = abs(\thisrow{m}), col sep = space,]{results/data-sum-strike/ad-euler-levels.csv};

	\addplot[mark = diamond, mark size = 3pt] 
	table[x = level, y expr = abs(\thisrow{m}), col sep = space,]{results/data-sum-strike/st-gam1-euler-levels.csv};

	\addplot[mark = o, mark size = 3pt] 
	table[x = level, y expr = abs(\thisrow{m}), col sep = space,]{results/data-sum-strike/st-gam2-euler-levels.csv};
	
	\draw(3.2, 0.000005) node[fill = white,]{$\Order{2^{-2\ell}}$};
	\addplot[domain = 2:6]{0.005*2^(-2*x)};
	
	\draw(4.7, 0.013) node[fill = white]{$\Order{2^{-\ell}}$};
	\addplot[domain = 2:6]{0.075*2^(-x)};
	\end{axis}
	\end{tikzpicture}
	\begin{tikzpicture}
	\begin{axis}[
	xlabel = $\ell$,
	ylabel = $E_\ell$,
	grid = both,
	ymode = log,
	]
	\addplot[mark = asterisk, mark size = 3pt] 
	table[x = level, y expr = abs(\thisrow{m}), col sep = space,]{results/data-sum-strike/ad-milstein-levels.csv};

	\addplot[mark = diamond, mark size = 3pt] 
	table[x = level, y expr = abs(\thisrow{m}), col sep = space,]{results/data-sum-strike/st-gam1-milstein-levels.csv};

	\addplot[mark = o, mark size = 3pt] 
	table[x = level, y expr = abs(\thisrow{m}), col sep = space,]{results/data-sum-strike/st-gam2-milstein-levels.csv};
	
	\draw(3.2, 0.000001) node[fill = white,]{$\Order{2^{-2\ell}}$};
	\addplot[domain = 2:6]{0.002*2^(-2*x)};
	
	\draw(4.7, 0.01) node[fill = white]{$\Order{2^{-\ell}}$};
	\addplot[domain = 2:6]{0.03*2^(-x)};
	\end{axis}
	\end{tikzpicture}
\end{center}
\caption{$W_\ell$ (top), $V_\ell$ (middle) and $E_\ell$ (bottom) vs $\ell$ for the one dimensional SDE problem using Euler-Maruyama (left) and Milstein (right) simulation of the underlying SDE. We consider non-adaptive samplers with $\gamma=1, 2$ and adaptive sampling with $r=1.95$ for the Euler scheme and $r=10$ for Milstein simulation.}
\label{fig:sde-1d-levels}
\end{figure}
 
For the  non-adaptive schemes with $\gamma = 2$ and the adaptive samplers, we compute $\mlmc$ for various error tolerances $\tol$. In \Cref{fig:sde-1d-mlmc}, we plot the total work (taken as overall number of SDE time-steps) times $\tol^2$ versus $\tol$, normalized by the true solution. For the non-adaptive, Euler-Maruyama sampler, we observe a rate close to $\tol^{-2.5}$ as predicted by \Cref{cor:st-comp-weak}. This is reduced to $\tol^{-2}\p*{\log\tol}^2$ using adaptive sampling with the Euler-Maruyama scheme as in \Cref{corr:ad-comp} for $q\to \infty$. Note that we observe the same rate without adaptive sampling when using the Milstein scheme, by \Cref{cor:st-comp-weak} since $\beta = 2\gamma$. The cost is slightly than for the adaptive Euler-Maruyama sampler, since the variance rate $V_\ell\lesssim 2^{-\gamma\ell}$ is observed without refining the samples beyond level $\ell$ at all.  However, we obtain the best results by combining the Milstein scheme with adaptive MLMC. In this case we observe complexity very close to $\tol^{-2}$ as in \Cref{corr:ad-comp}.\\

\begin{figure}[h]
	\centering 
	\begin{center}
	\begin{minipage}{0.4\linewidth}
		\begin{center}
		Euler-Maruyama\\
		\begin{tabular}{cc}
			\ref*{plot:gam2-e-mlmc} $\gamma = 2$ & \ref*{plot:ad-e-mlmc} $r = 1.95$
		\end{tabular}
		\end{center}
	\end{minipage}
	\begin{minipage}{0.4\linewidth}
		\begin{center}
		Milstein \\
		\begin{tabular}{ccc}
			\ref*{plot:gam2-m-mlmc} $\gamma = 2$ &  \ref*{plot:ad-m-mlmc} $r = 10$
		\end{tabular}
		\end{center}
	\end{minipage}
	\end{center}
	\begin{tikzpicture}[trim axis left, trim axis right]
	\begin{axis}[
	xlabel = Normalised $\tol$,
	ylabel = $\textnormal{Total Work}\times\tol^2$,
	grid = both,
	ymode = log,
	xmode = log,
	legend pos = north east,
	legend columns = 2,
	ymax = 700,
	width = \linewidth,
	height = 0.4\linewidth,
	]
	\addplot[style = dashed, domain = 0.0001:0.1,]{0.2*ln(x)*ln(x)};
	\addlegendentry{$\tol^{-2}\p*{\log\tol}^2$}
	
	\addplot[style = densely dotted, domain = 0.0001:0.1,]{5*x^(-0.5)};
	\addlegendentry{$\tol^{-2.5}$}
	
	\addplot+[mark = asterisk, mark size = 3pt,black] 
	table[x expr = \thisrow{tol}/0.025, y expr = \thisrow{cost}*\thisrow{tol}*\thisrow{tol}, col sep = space,]{results/data-sum-strike/ad-euler-mlmc.csv};
	\label{plot:ad-e-mlmc}

	\addplot+[mark = o, mark size = 3pt,black] 
	table[x expr = \thisrow{tol}/0.025, y expr = \thisrow{cost}*\thisrow{tol}*\thisrow{tol}, col sep = space,]{results/data-sum-strike/st-gam2-euler-mlmc.csv};
	\label{plot:gam2-e-mlmc}

	\addplot+[style = solid, mark = square, mark size = 3pt,black] 
	table[x expr = \thisrow{tol}/0.025, y expr = \thisrow{cost}*\thisrow{tol}*\thisrow{tol}, col sep = space,]{results/data-sum-strike/ad-milstein-mlmc.csv};
	\label{plot:ad-m-mlmc}
	
	\addplot[ style = solid, mark size = 3pt,black,mark = *,] 
	table[x expr = \thisrow{tol}/0.025, y expr = \thisrow{cost}*\thisrow{tol}*\thisrow{tol}, col sep = space,]{results/data-sum-strike/st-gam2-milstein-mlmc.csv};
	\label{plot:gam2-m-mlmc}
	\end{axis}
	\end{tikzpicture}
	\caption{The total work times $\tol^2$  of MLMC for the one-dimensional SDE problem versus $\tol$, normalised by the true value of the solution. Results are shown for both Euler-Maruyama and Milstein simulation of the underlying SDE. For each method we show results for non-adaptive sampling with $\gamma = 2$ and adaptive sampling with $r = 1.95$ for the Euler scheme and $r = 10$ for Milstein simulation.  }
	\label{fig:sde-1d-mlmc}
\end{figure}
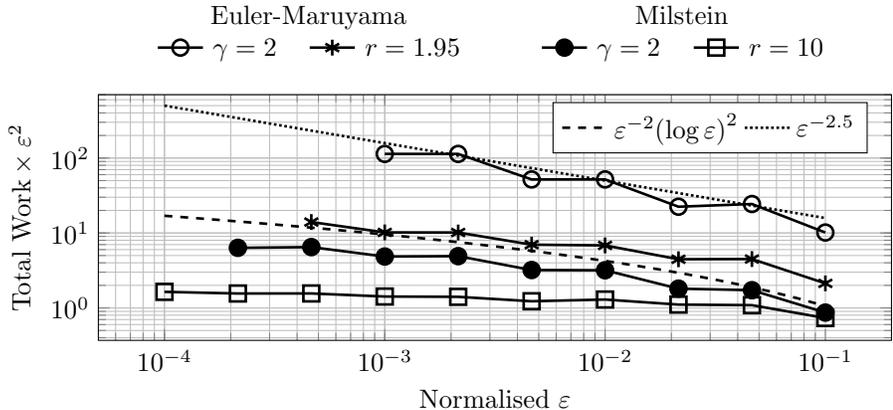
 While the Milstein scheme offers better performance for the one-dimensional problem, this method becomes unfeasible in large dimensions, requiring approximation of double It\^{o} integrals at each step. On the other hand, we still observe a significant improvement by combining adaptive sampling with the Euler-Maruyama scheme, as seen in the one-dimensional problem. To illustrate how this performance translates to higher dimensional problems we consider the case $d=10$, with correlation coefficient $\rho = 0.2$. Euler-Maruyama simulation of $g$ is used in non-adaptive MLMC for $\gamma = 1,2$ and adaptive sampling with $\gamma = 1$ and $r=1.95$. The parameter $K$ is again tuned so that $\E*{\H{g}}\approx 0.025$. The values of $W_\ell, V_\ell, E_\ell$ for each method are plotted against $\ell$ in \Cref{fig:d10-results}. We observe a slight increase to each term.  However, the rates of each parameter are all equivalent to those seen before, and the complexity of MLMC is unaffected by the increased dimensionality. To emphasize this point, \Cref{fig:d10-results} also displays the total work $(\times \tol^2)$ against $\tol$. In particular, we again observe $\tol^{-2}\p*{\log\tol}^2$ complexity for the adaptive sampler, as opposed to $\tol^{-2.5}$ for the non-adaptive samplers.
\begin{figure}[h]
	\centering 
	\begin{tabular}{ccc}
		\ref*{plot:d10-gam1} $\gamma = 1$ & \ref*{plot:d10-gam2} $\gamma = 2$ & \ref*{plot:d10-ad} Adaptive
	\end{tabular}\\
	\begin{tabular}{cc}
		\ref*{plot:d10-log} $\tol^{-2}\p*{\log\tol}^2$ & \ref*{plot:d10-eps}  $\tol^{-2.5}$
	\end{tabular}\\
	\begin{tikzpicture}
	\begin{axis}[
	xlabel = $\ell$,
	ylabel = $W_\ell$,
	grid = both,
	ymode = log,
	]
	\addplot+[mark = asterisk, mark size = 3pt, black,] 
	table[x = level, y expr = \thisrow{cost}/\thisrow{M}, col sep = space,]{results/data-sum-strike/ad-d10-euler-levels.csv};
	\label{plot:d10-ad}

	\addplot+[mark = diamond, mark size = 3pt, black] 
	table[x = level, y expr = \thisrow{cost}/\thisrow{M}, col sep = space,]{results/data-sum-strike/st-gam1-d10-euler-levels.csv};
	\label{plot:d10-gam1}

	\addplot+[mark = o, mark size = 3pt, black] 
	table[x = level, y expr = \thisrow{cost}/\thisrow{M}, col sep = space,]{results/data-sum-strike/st-gam2-d10-euler-levels.csv};
	\label{plot:d10-gam2}

	\end{axis}
	\end{tikzpicture}
	\begin{tikzpicture}
	\begin{axis}[
	xlabel = $\ell$,
	ylabel = $V_\ell$,
	grid = both,
	ymode = log,
	]
	\addplot+[mark = asterisk, mark size = 3pt, black,] 
	table[x = level, y expr = \thisrow{V}, col sep = space,]{results/data-sum-strike/ad-d10-euler-levels.csv};

	\addplot+[mark = diamond, mark size = 3pt, black] 
	table[x = level, y expr = \thisrow{V}, col sep = space,]{results/data-sum-strike/st-gam1-d10-euler-levels.csv};

	\addplot+[mark = o, mark size = 3pt, black] 
	table[x = level, y expr = \thisrow{V}, col sep = space,]{results/data-sum-strike/st-gam2-d10-euler-levels.csv};
	
	\draw(3.2, 0.0003) node[fill = white,]{$\Order{2^{-\ell}}$};
	\addplot[domain = 2:6]{0.01*2^(-x)};
	
	\draw(4.7, 0.012) node[fill = white]{$\Order{2^{-\ell/2}}$};
	\addplot[domain = 2:6]{0.03*2^(-x/2)};
	\end{axis}
	\end{tikzpicture}\\
	\begin{tikzpicture}
	\begin{axis}[
	xlabel = $\ell$,
	ylabel = $E_\ell$,
	grid = both,
	ymode = log,
	]
	\addplot[mark = asterisk, mark size = 3pt] 
	table[x = level, y expr = abs(\thisrow{m}), col sep = space,]{results/data-sum-strike/ad-d10-euler-levels.csv};

	\addplot[mark = diamond, mark size = 3pt] 
	table[x = level, y expr = abs(\thisrow{m}), col sep = space,]{results/data-sum-strike/st-gam1-d10-euler-levels.csv};

	\addplot[mark = o, mark size = 3pt] 
	table[x = level, y expr = abs(\thisrow{m}), col sep = space,]{results/data-sum-strike/st-gam2-d10-euler-levels.csv};
	
	\draw(3.2, 0.000005) node[fill = white,]{$\Order{2^{-2\ell}}$};
	\addplot[domain = 2:6]{0.005*2^(-2*x)};
	
	\draw(4.7, 0.013) node[fill = white]{$\Order{2^{-\ell}}$};
	\addplot[domain = 2:6]{0.075*2^(-x)};
	\end{axis}
	\end{tikzpicture}
	\begin{tikzpicture}
	\begin{axis}[
	xlabel = Normalised $\tol$,
	ylabel = $\textnormal{Total Work}\times\tol^2$,
	grid = both,
	ymode = log,
	xmode = log,
	legend pos = north east,
	legend columns = 2,
	ymax = 200,
	]
	\addplot[style = dashed, domain = 0.001:0.1,]{0.5*ln(x)*ln(x)};
	\label{plot:d10-log}
	
	\addplot[style = densely dotted, domain = 0.001:0.1,]{5*x^(-0.5)};
	\label{plot:d10-eps}
	
	\addplot+[mark = asterisk, mark size = 3pt,black] 
	table[x expr = \thisrow{tol}/0.025, y expr = \thisrow{cost}*\thisrow{tol}*\thisrow{tol}, col sep = space,]{results/data-sum-strike/ad-d10-euler-mlmc.csv};

	\addplot+[mark = diamond, mark size = 3pt,black] 
	table[x expr = \thisrow{tol}/0.025, y expr = \thisrow{cost}*\thisrow{tol}*\thisrow{tol}, col sep = space,]{results/data-sum-strike/st-gam1-d10-euler-mlmc.csv};
	
	\addplot+[mark = o, mark size = 3pt,black] 
	table[x expr = \thisrow{tol}/0.025, y expr = \thisrow{cost}*\thisrow{tol}*\thisrow{tol}, col sep = space,]{results/data-sum-strike/st-gam2-d10-euler-mlmc.csv};

	\end{axis}
	\end{tikzpicture}
	\caption{Results for 10-dimensional digital option: Expected work per level $W_\ell$ (top left) taken as expected number of time-steps required to sample $g_\ell$ and $g_{\ell-1}$,  $V_\ell$(top right) and $E_\ell$ (bottom left) versus $\ell$. The total work of MLMC times $\tol^2$ against $\tol$, normalised by the true value of the solution is shown bottom right. Results are  given for non-adaptive schemes with $\gamma =1, 2$ and the adaptive scheme with $r=1.95$. }
	\label{fig:d10-results}
\end{figure}
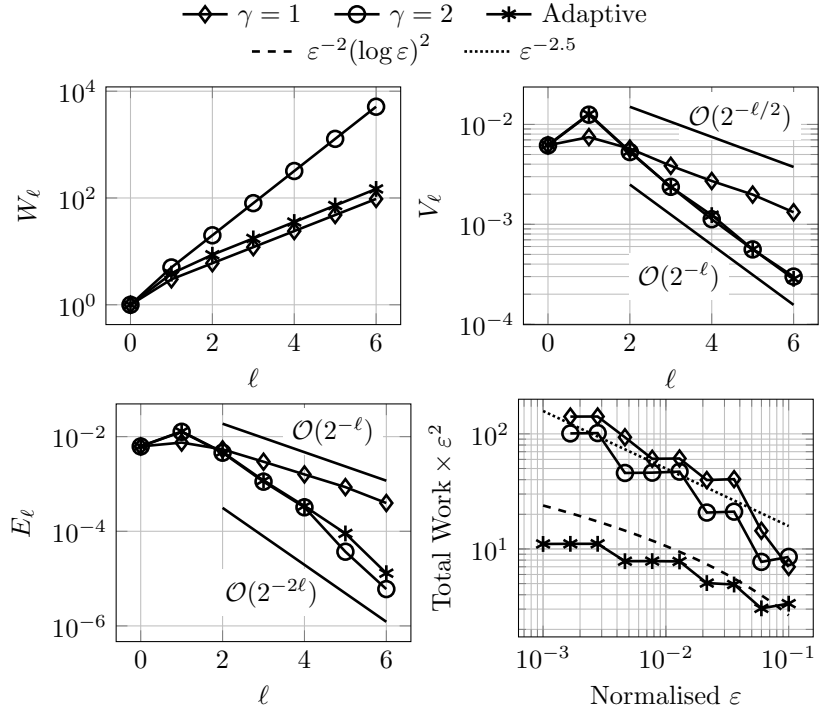

  \section{Conclusion}
We presented an efficient, general, MLMC framework for computing \eqref{eqn:heaviside-problem}. The inherent discontinuity in the problem leads to high complexities for standard MLMC methods. We are able to improve the performance of MLMC using adaptive sampling based on the methods for nested simulation in \cite{GilesHajiAli:2018}. The approach used is applicable to a wide class of problems and is often able to recover the canonical $\tol^{-2}$ MLMC complexity. The theory is supported by numerical experiments for nested simulation and SDEs. It is straightforward to extend the methods considered here to compute expectations of discontinuous functionals other than $\I{G\in\Omega}$ or $\H{g}$. For example, in barrier option pricing, the payoff can be written as a product of a smooth/Lipschitz function with an indicator function. We will consider applications to other financial derivatives and risk measures in future work.

The next step in this research is  to investigate the use of Multilevel Quasi-Monte Carlo methods to reduce the cost even further, potentially to an $\Order{\tol^{-1}}$ operation. 

\subsubsection*{Acknowledgments}
We wish to acknowledge the helpful input and feedback received from Michael B. Giles throughout the development of this paper.

A-L. Haji-Ali was supported by a Sabbatical Grant from the Royal Society of Edinburgh. 

J. Spence was supported by EPSRC grant EP/S023291/1.

\bibliographystyle{abbrv}

\appendix
\normalsize
\section{Discussion of \Cref{assumpt:weak-err-adapt}}\label{app:weakerr}
This appendix discusses the necessity of \Cref{assumpt:weak-err-adapt} to observe better convergence rates for $E_\ell$ due to adaptive sampling. Evidence is given in the form of a numerical experiment when \Cref{assumpt:weak-err-adapt} is false. In particular, we consider \eqref{eqn:heaviside-problem} where $g \sim \mathcal{N}(-\mu, 1)$ for $\mu >0$ chosen such that $\prob{g>0} = 0.025$. Approximations $g_\ell$ are artificially sampled through 
\[
	g_\ell = g + 2^{-\ell\gamma/2}\p*{2^{-\ell\gamma/2} + \zeta^2 - 1}, \quad \zeta\sim\mathcal{N}(0,1). 
\]
We assume an artificial cost of $2^{\gamma\ell}$ in sampling $g_\ell$. It follows that \Cref{assumpt:q-moments,assumpt:delta-bound} hold for $\beta = \gamma$, $\sigma_\ell \equiv \sigma$ constant and any $q<\infty$.  From \eqref{eqn:Z_ell},
\[
Z_\ell = \frac{2^{-\ell\gamma/2} + \zeta^2 - 1}{\sigma},
\]
and so \Cref{assumpt:weak-err-nondiscont} holds for $\beta=\gamma$.  However, \Cref{assumpt:weak-err-adapt} is false since $\prob{Z_\ell>0} \to 1$ as $\ell\to\infty$. Thus, the hypothesis of \Cref{prop:weak-err-ad} is false.\\

\Cref{fig:weakerr-results} plots $E_\ell$ as in \eqref{weak-err} for non-adaptive sampling with $\gamma=1,2$ and for adaptive sampling with $r=1.95, \gamma=\theta=c=1$ and $\sigma_\ell = \sigma = \sqrt{3}$. We use the same sample of $\zeta$ for the fine and coarse levels in each MLMC sample, and when adaptively refining samples. For all methods we see $E_\ell \propto 2^{-\gamma\ell}$. Since $\beta=\gamma$ this agrees with \Cref{thm:standard-bias} for the non-adaptive samplers. However, \Cref{prop:weak-err-ad} concludes that $E_\ell \lesssim 2^{-2\ell}$ for the adaptive sampler in this setup, in contrast to the $\Order{2^{-\ell}}$ convergence seen here.

\begin{figure}[h]
	\centering 
	\begin{tabular}{ccc}
		\ref*{plot:d10-gam1} $\gamma = 1$ & \ref*{plot:d10-gam2} $\gamma = 2$ & \ref*{plot:d10-ad} Adaptive
	\end{tabular}\\
	\begin{tikzpicture}[trim axis left]
	\begin{axis}[
	xlabel = $\ell$,
	ylabel = $E_\ell$,
	grid = both,
	ymode = log,
	]
	\addplot[mark = asterisk, mark size = 3pt] 
	table[x = level, y expr = abs(\thisrow{m}), col sep = space,]{results/data-basic/ad-levels.csv};

	\addplot[mark = diamond, mark size = 3pt] 
	table[x = level, y expr = abs(\thisrow{m}), col sep = space,]{results/data-basic/st-gam1-levels.csv};

	\addplot[mark = o, mark size = 3pt] 
	table[x = level, y expr = abs(\thisrow{m}), col sep = space,]{results/data-basic/st-gam2-levels.csv};
	
	\draw(3.2, 0.00003) node[fill = white,]{$\Order{2^{-2\ell}}$};
	\addplot[domain = 2:6]{0.04*2^(-2*x)};
	
	\draw(4.7, 0.15) node[fill = white]{$\Order{2^{-\ell}}$};
	\addplot[domain = 2:6]{0.75*2^(-x)};
	\end{axis}
	\end{tikzpicture}
	\caption{$E_\ell$ versus $\ell$ for the artificial problem used as evidence of worse weak error rates for adaptive sampling when \Cref{assumpt:weak-err-adapt} is false. }
	\label{fig:weakerr-results}
\end{figure}
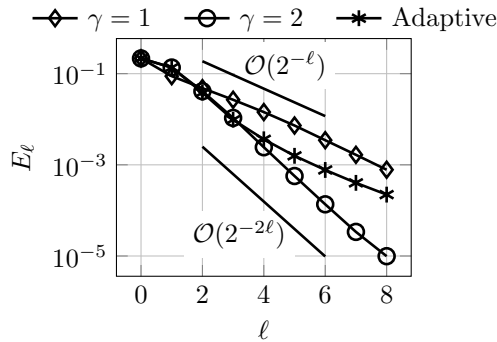

  \section{Different values of $\sigma_\ell$}\label{app:constsigma}
In \Cref{sect:nest-ex} we assumed $\sigma_\ell^2$ was the sample conditional variance of $X$ given $Y$. However, the other examples considered all use constant values $\sigma_\ell\equiv \sigma$. In this appendix we discuss other choices of $\sigma_\ell$ for the nested simulation problem and the impact on the work of MLMC. One option is to take $\sigma_\ell^2 = \var{X|Y}$. However, this information is likely unavailable for all practical applications. Thus, we consider instead the approximation $\sigma_\ell \equiv \sigma$, for some constant $\sigma>0$. For \Cref{assumpt:q-moments} to hold for $\sigma_\ell\equiv\sigma$, we now require bounded moments of $g-g_\ell$, opposed to the self normalized process appearing in \eqref{eqn:Z-nest}, which is a more restrictive condition. \\

We present results using adaptive sampling for the model problem presented in \Cref{sect:nest-ex}. Specifically, we estimate the total work of MLMC with several error tolerances $\tol$ and optimal starting levels as in \Cref{sect:nest-ex} except constant $\sigma_\ell = \sigma$. The total work required with fixed $\sigma$ divided by the work when using the sample standard deviation is shown in \Cref{fig:const-sig}. The solid markers show the value $\sigma^2 = \E{\var{X|Y}}$. When $\sigma \to 0$, the term $\abs{\delta_\ell} $ in \Cref{alg:ad-samp} tends to $\infty$ and we instead use deterministic sampling with $N_\ell = N_02^{\ell}$ inner samples per level in the limiting case. Conversely, when $\sigma \to \infty$, $\abs{\delta_\ell}\to 0$ and the adaptive algorithm reverts to deterministic sampling with $N_\ell = N_02^{2\ell}$ inner samples per level. This leads to expensive pre-asymptotic regimes for large and small $\sigma$ and we observe worse performance as $\tol$ decreases.  That the work is typically lower for large $\sigma$ opposed to small $\sigma$ is consistent with results showing MLMC is more effective when the approximations are refined by a factor of around $7$ per level in this application ($W_\ell \propto 7^\ell$) \cite{Giles2008MLMC}. The only value of $\sigma$ for which we consistently observe equal performance using constant $\sigma_\ell$, opposed to the sample variance, is $\sigma^2 = \E{\var{X|Y}}$. MLMC actually has slightly lower cost for constant $\sigma_\ell$ in this instance, likely due to statistical errors in the sample variance impacting the refinement of certain samples, whereas fixing $\sigma^2 = \E{\var{X|Y}}$ refines samples enough on average to observe the benefits of adaptive sampling. To draw further conclusions, we require more rigorous justification of \Cref{assumpt:q-moments,assumpt:delta-bound} for this problem. 
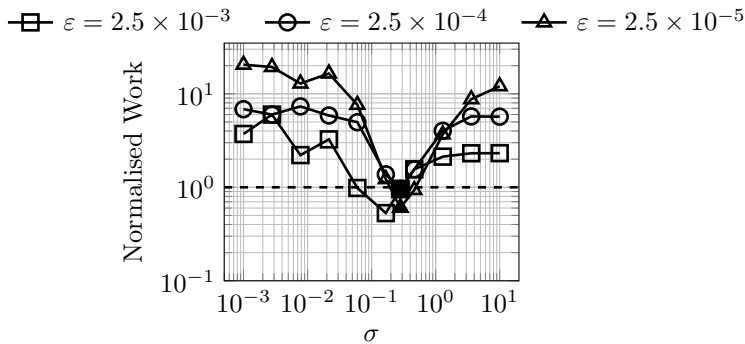
\begin{figure}[h]
	\centering 
	\begin{tabular}{ccc}
		\ref*{plot:weaksmall} $\tol = 2.5\times10^{-3}$ & \ref*{plot:weakmid} $\tol = 2.5\times10^{-4}$ & \ref*{plot:weakbig} $\tol = 2.5\times10^{-5}$  
	\end{tabular}

	\begin{tikzpicture}[trim axis left]
	\begin{axis}[
	xlabel = $\sigma$,
	ylabel = Normalised Work,
	grid = both,
	xmode = log,
	ymode = log,
	ymin = 0.1,
	xmin = 0.0005,
	xmax = 20
	]
	\addplot[mark = square, mark size = 3pt] 
	table[x = sigma, y = worksmall, col sep = comma,]{results/data-nest-mp/const-sig.csv};
	\label{plot:weaksmall}
	
	\addplot[mark = square*, mark size = 3pt]coordinates{(0.281424945589406, 0.958857085609794)};
	
	\addplot[mark = o, mark size = 3pt] 
	table[x = sigma, y = workmid, col sep = comma,]{results/data-nest-mp/const-sig.csv};
	\label{plot:weakmid}
	
	\addplot[mark = *, mark size = 3pt]coordinates{(0.281424945589406, 0.958857085609794)};
	
	\addplot[mark = triangle, mark size = 3pt] 
	table[x = sigma, y = workbig, col sep = comma,]{results/data-nest-mp/const-sig.csv};
	\label{plot:weakbig}
	
	\addplot[mark = triangle*, mark size = 3pt]coordinates{(0.281424945589406, 0.604130247633887)};
	
	\addplot[domain=0.0005:20, dashed]{1};
	\end{axis}
	\end{tikzpicture}
	
	\caption{The work required for the model problem in \Cref{sect:nest-ex} using adaptive MLMC with constant $\sigma_\ell = \sigma$, normalised by the work when $\sigma_\ell$ is the sample variance. The solid markers show $\sigma = \sqrt{\E*{\var{X|Y}}}\approx 0.28$.}
	\label{fig:const-sig}
\end{figure}

  \end{document}